\documentclass[draft,12pt,reqno,a4paper]{amsart}

\usepackage[margin=3cm,footskip=1cm]{geometry}

\usepackage{amssymb} 
\usepackage{amsmath} 
\usepackage{amsthm} 
\usepackage{mathtools}

\usepackage{enumerate} 
\usepackage[latin1]{inputenc}
\usepackage[shortlabels]{enumitem}
\usepackage{color}
\usepackage{graphicx} 
\usepackage{verbatim}

\DeclarePairedDelimiter\set{\{}{\}}

\newcommand{\N}{{\mathbb{N}}} 
\newcommand{\R}{{\mathbb{R}}} 
\newcommand{\C}{{\mathbb{C}}}

 \renewcommand{\c}{{\rm c}}

\newcommand\parb[2][]{#1 \big ( #2#1\big )}

 \newcommand{\vB}{{\mathcal B}}
 \newcommand{\vD}{{\mathcal D}}

 \newcommand{\vL}{{\mathcal L}}

\theoremstyle{plain}
\newtheorem{thm}{Theorem}[section]
\newtheorem{proposition}[thm]{Proposition}
\newtheorem{lemma}[thm]{Lemma} \newtheorem{corollary}[thm]{Corollary}
\theoremstyle{definition}

 \newtheorem{cond}[thm]{Condition}
 \newtheorem{remark}[thm]{Remark}

 \newtheorem*{remarks*}{Remarks}
\newtheorem*{remark*}{Remark}

\numberwithin{equation}{section}

\title{Spectral theory for  $1$-body Stark operators}

\thanks{
K.Ito is supported by JSPS KAKENHI grant nr.\ 17K05325.\,
E.S. is  supported by the Research Institute for Mathematical Sciences, a Joint
Usage/Research Center located in Kyoto University, and by DFF grant nr.\ 4181-00042. 
} 

\author{T. Adachi}
\address[T. Adachi]{Graduate School of Human and Environmental Studies,
        Kyoto University, Kyoto, Japan}
\email{adachi@math.h.kyoto-u.ac.jp}

\author{K. Itakura}
\address[K. Itakura]{Department of Mathematics, Kobe University\\
1-1 Rokkodai, Nada, Kobe, 657-8501, Japan}
\email{itakura@math.kobe-u.ac.jp}

\author{K. Ito}
\address[K. Ito]{Graduate School of Mathematical Sciences, The University of Tokyo\\
3-8-1 Komaba, Meguro-ku, Tokyo 153-8914, Japan}
\email{ito@ms.u-tokyo.ac.jp}

\author{E. Skibsted} \address[E. Skibsted]{Institut for Matematiske Fag\\
Aarhus Universitet\\ Ny Munkegade 8000 Aarhus C, Denmark}
\email{skibsted@math.au.dk}

\date{\today}
\begin{document}
\begin{abstract}
We investigate  spectral theory for  a one-body Stark Hamiltonian under
minimum regularity and decay conditions on the potential (actually allowing sub-linear growth at
infinity). Our results include 
Rellich's theorem, the limiting absorption principle,
 radiation condition bounds and Sommerfeld's uniqueness, and most of 
 these 
 are stated and proved in sharp form employing  Besov-type spaces.
For the proofs we adopt a commutator scheme 
by Ito--Skibsted \cite{IS}. 
A feature of the paper is a particular choice of an escape function related to  parabolic coordinates, 
which conforms  well with classical mechanics for  the Stark Hamiltonian.
The whole setting of the paper, such as 
the conjugate operator and the Besov-type spaces, 
is generated by this single escape function. We apply our  results  in the sequel paper \cite{AIIS2}.
\end{abstract}

\allowdisplaybreaks

\maketitle
\tableofcontents

\section{Introduction}

In this paper we investigate  spectral theory for a \emph{perturbed} Stark Hamiltonian
on the Euclidean space of dimension $d\geq 2$. 
Let us split the space variable of $\mathbb R^d$ as 
$(x,y)=(x,y_2,\ldots,y_d)\in \mathbb R\times \mathbb R^{d-1}$ 
and apply the Stark field in the positive $x$-direction.
The \emph{free} Stark Hamiltonian is given by 
$$H_0=\tfrac12(p_x^2+p_y^2)-x=\tfrac12(p_x^2+p_{y_2}^2+\dots+p_{y_d}^2)-x;\quad 
p=-\mathrm i\partial.$$
We perturb it, and consider 
$$H=H_0+q,$$
where $q$ is a multiplication operator 
by a real-valued function $q\in L^2_{\mathrm{loc}}(\mathbb R^d)$ (with
more regularity for $d\geq 4$).
We assume that $q$ has a weaker growth rate at infinity than the Stark field in some appropriate sense.

We are going to present a  spectral analysis of $H$, 
and our  main results are Rellich's theorem, LAP (\emph{Limiting Absorption Principle}),
  radiation condition bounds and Sommerfeld's uniqueness result.
The precise statements will be given in Section~\ref{sec:Setting and results}. 
These results are known for perturbations of the free Laplacian but
seem to a substantial  degree  to be missing
for Stark Hamiltonians even for  the $1$-body case,
definitely   in the  sharp form as derived here. 
We refer to \cite{AH,H,Y1,Y2,W} for directly related spectral results 
for $1$-body Stark Hamiltonians,
and to \cite{Si,Ta1,Ta2,Sk,HMS1} for $N$-body generalizations.

The stationary scattering theory for Stark Hamiltonians is not fully
developed,
although asymptotic completeness of the time-dependent wave operators
was established  long ago, even for 
  $N$-body Stark Hamiltonians 
 \cite{AT1,AT2,HMS2}. In our sequel paper \cite{AIIS2} we study the
 stationary scattering theory for the one-body problem for a  more
 restrictive class of potentials than considered here. In particular we
 shall derive detailed information on the scattering matrix using
  results from this paper (in particular Sommerfeld's uniqueness result).

We prove our results using the commutator scheme 
developed by Ito--Skibsted \cite{IS}, 
and  the choice of an escape function $f$, given by \eqref{eq:181215},
is a  novelty of the present paper. 
In the scheme of \cite{IS} 
an escape function plays an central role, 
generating the `conjugate operator' $A$ and the associated Besov spaces 
$\mathcal B,\mathcal B^*$ and $\mathcal B^*_0$, 
see \eqref{eq:def-A} and \eqref{eq:1812156}, respectively. 
Our escape function $f$ is intimately related to  parabolic coordinates, 
and it has several appealing  features from a classical mechanical viewpoint. 
As far as we know, it seems to be the first time  such $f$ is employed  
in commutator theory of Stark Hamiltonians, although   superficially 
there is some similarity  with a construction of \cite{HMS1} (for
example). We  refer to \cite{I1,I2} 
for applications of the scheme of  \cite{IS} 
to repulsive one-body Hamiltonians.

It is well known that  Mourre theory \cite{Mo}, under conditions on
the potential, yields LAP for Stark Hamiltonians. Although we call $A$
a `conjugate operator' it does not conform with the notion of
conjugate operator of \cite{Mo}, in fact our $A$ is bounded relatively
to the Hamiltonian (like the $A$ of \cite{IS}).
Nevertheless the commutator $\mathrm i[H,A]$ possesses some positivity
justifying our terminology (of course this positivity is very weak 
 and spatially non-uniform), 
see Lemma~\ref{lemma:commutator-calculation} with $\Theta\equiv 1$.

This paper is organized as follows. 
In Section~\ref{sec:Setting and results} we present all the assumptions
and all the main results of the paper.
Section~\ref{sec:Conjugate operator} is a short preliminary for proofs 
of the main results, where we implement a commutator computation.
Section~\ref{sec:180629} is devoted to the proof of Rellich's theorem,
and Section~\ref{sec:181010} to that of LAP bounds.
In Section~\ref{sec:181207} we first prove 
the radiation condition bounds for {complex} values of the spectral parameter,
and then we prove LAP, the radiation condition bounds for {real} values of the spectral parameter
and Sommerfeld's uniqueness result.

\section{Setting and results}\label{sec:Setting and results}

In this section we precisely formulate our setting, 
and then state all the main results of the paper. 

Throughout the paper we fix our \emph{escape function}: 
\begin{align}
f=f(x,y)
=\chi(r+x)+\bigl[1-\chi(r+x)\bigr](r+x)^{1/2}
;\quad
r=(x^2+y^2)^{1/2},
\label{eq:181215}
\end{align}
where $\chi \in C^\infty(\mathbb R)$ is a real-valued and smooth cut-off function satisfying 
\begin{equation}\label{eq:chidef}
\chi(s)
=\left\{\begin{array}{ll}
1 &\mbox{for } s \le 1, \\
0 &\mbox{for } s \ge 2,
\end{array}
\right.
\quad
\tfrac{\mathrm d}{\mathrm ds}\chi=\chi'\le 0.
\end{equation}
Such choice of $f$ is stimulated by \cite{HMS1}, 
but ours is completely different from theirs. 
One particular difference is that the level surfaces of $f$ are paraboloids, 
while those of \cite{HMS1} are distorted spheres. 
Actually $r+x$ is exactly one of the components of a choice of  parabolic coordinates in $\mathbb R^d$.
Thus the gradient vector field of $f$ is tangent to another family of 
paraboloids of the converse direction,
which asymptotically conforms better with the classical orbits of the Stark Hamiltonian.
It is well known that in the parabolic coordinates 
the method of separation of variables works for the free Stark
Hamiltonian, see e.g.\ \cite{Ti},
however our motivation is different.

Letting $F(S)$ be the characteristic function of a general subset $S\subseteq \mathbb R^d$,
we set
\begin{align}
F_n&=F\bigl(\bigl\{ (x,y)\in \mathbb R^d\,\big|\,2^n\le f(x,y)<2^{n+1}\bigr\} \bigr)
\quad \text{for }n\in\mathbb N_0=\mathbb N\cup\{0\}.
\label{eq:190316}
\end{align}
Then define the \emph{Besov spaces $\vB$, $\mathcal B^*$ and $\mathcal B^*_0$ associated with $f$} as 
\begin{align}
\begin{split}
\vB&=
\Bigl\{\psi\in L^2_{\mathrm{loc}}(\mathbb R^d)\,\Big|\, 
\sum_{n\in\N_0}2^{n/2}\|F_n\psi\|_{L^2}<\infty\Bigr\},
\\
\mathcal B^*&=
\Bigl\{\psi\in L^2_{\mathrm{loc}}(\mathbb R^d)\,\Big|\, 
\sup_{n\in\mathbb N_0}2^{-n/2}\|F_n\psi\|_{L^2}<\infty\Bigr\},
\\
\mathcal B^*_0
&=
\Bigl\{\psi\in \mathcal B^*\,\Big|\, \lim_{n\to\infty}2^{-n/2}\|F_n\psi\|_{L^2}=0\Bigr\}.
\end{split}
\label{eq:1812156}
\end{align}
Note that these are Banach spaces with respect to the norms
\begin{align*}
\|\psi\|_\vB
&=\sum_{n\in\N_0}2^{n/2}\|F_n\psi\|_{L^2},
\quad
\|\psi\|_{\mathcal B^*}=\|\psi\|_{\mathcal B^*_0}
=\sup_{n\in\mathbb N_0}2^{-n/2}\|F_n\psi\|_{L^2}.
\end{align*}
Note also that, if we introduce the \emph{$f$-weighted $L^2$-spaces of order $s\in\mathbb R$} as 
$$L^2_s=f^{-s}L^2,$$
then for any $s>1/2$ the following proper inclusions hold:
\begin{align}\label{eq:inclusion-relations}
L^2_s \subsetneq \vB \subsetneq L^2_{1/2} \subsetneq L^2 
\subsetneq L^2_{-1/2} \subsetneq \vB_0^* \subsetneq \vB^* \subsetneq L^2_{-s}.
\end{align}

Our first theorem is \emph{Rellich's theorem}, which asserts absence of 
generalized eigenfunctions in $\mathcal B^*_0$ for $H$  under the
following conditions on $q$,  which by \cite[Theorems X.29 and
X.38]{RS} are
sufficient for 
essentially self-adjointness of $H$ on $C^\infty_\c(\R^d)$.
Define a differential operator $\partial^f$ in direction to $\mathop{\mathrm{grad}} f$ as 
\begin{align}
\partial^f=(\partial f)\partial
=(\partial_xf)\partial_x+(\partial_yf)\partial_y.
\label{eq:180925}
\end{align}
 Let 
$H^s=H^s(\mathbb R^d)$ denote  the standard Sobolev space.
\begin{cond}\label{cond:1803131348}
There exists a splitting $q=q_1+q_2+q_3$ 
by real-valued measurable functions $q_j$, $j=1,2,3$, such that 
for some $\rho,C>0$:
\begin{enumerate}[(1)]
\item
$q_1$ is continuously differentiable, and satisfies for any $(x,y)\in \mathbb R^d$
\begin{align}
|q_1(x,y)|\le Cf^{-\rho}r,\quad 
\partial^f q_1(x,y)\le Cf^{-1-\rho}
;
\label{eq:1812151221}
\end{align}
\item
$q_2$ satisfies for any $(x,y)\in \mathbb R^d$ 
\begin{align*}
|q_2(x,y)|\le Cf^{-1-\rho}
;
\end{align*}
\item
$q_3$ is compactly supported, and 
the associated multiplication operator is compact as $H^2\to H^{0}$.
\end{enumerate}
\end{cond}
\begin{remark}\label{remark:fbndg}
  Note that $f^2\le 2r$ holds true outside some compact subset of
  $\mathbb R^d$, but the converse $cr\le f^2$ is false.  Note also
  that $|\partial f|^2=\tfrac12r^{-1}$ for $f$ large, cf.\
  \eqref{eq:180816}, and that in general $|\partial f|^2\leq Cr^{-1}$.
\end{remark}

\begin{cond}\label{cond:smooth2wea3n2} 
If $\phi\in L^2_{\mathrm{loc}}(\mathbb R^d)$ satisfies 
\begin{enumerate}
\item
$(H-\lambda)\phi=0$ for some $\lambda\in\mathbb R$ in the distributional sense,
\item
$\phi=0$ on a non-empty open subset of $\mathbb R^d$, 
\end{enumerate}
then $\phi=0$ on $\mathbb R^d$.
\end{cond}

\begin{remark*}
The property required in Condition~\ref{cond:smooth2wea3n2} is called the 
\emph{unique continuation property}. 
We consider it as a rather independent topic and will not discuss it in this paper, 
only referring to \cite{Wo} for some criteria. 
One sufficient condition in our setting is that, quite roughly speaking, 
`singularities' of $q_3$ do not separate 
the space $\mathbb R^d$ into plural components. 
In particular, if $q_3\equiv 0$, Condition~\ref{cond:smooth2wea3n2} holds automatically.
\end{remark*}

Using the function $\chi$ from \eqref{eq:chidef},
we define \emph{smooth cut-off functions} 
$\chi_m,\bar\chi_m,\chi_{m,n}\in C^\infty(\R^d)$ for $m,n\in\mathbb N_0$ as 
\begin{equation}\label{eq:chimn}
\chi_m=\chi(f/2^m), \quad \bar \chi_m=1-\chi_m, \quad \chi_{m,n}=\bar\chi_m\chi_n.
\end{equation}

\begin{thm}\label{thm:priori-decay-b_0} 
Assume Conditions~\ref{cond:1803131348} and \ref{cond:smooth2wea3n2}. Let $\lambda\in\mathbb R$.
If $\phi\in L^2_{\mathrm{loc}}(\mathbb R^d)$ satisfies  
\begin{enumerate}[(1)]
\item
$(H-\lambda)\phi=0$ in the distributional sense,
\item
$\bar\chi_{m_0}\phi\in \mathcal B_0^*$ for some $m_0\in\mathbb N_0$,
\end{enumerate}
then $\phi=0$ on $\mathbb R^d$.
\end{thm} 

\begin{remark*} We show in \cite{AIIS2} that under more restrictive
  conditions on $q$ there are lots of generalized eigenfunctions in
  $\mathcal B^*$, see  Remark \ref{rem:1812248} below. Thus we can consider
  Theorem~\ref{thm:priori-decay-b_0} to be  optimal.
\end{remark*}

The  proof of Theorem~\ref{thm:priori-decay-b_0} will be given in Section~\ref{sec:180629}.
The following corollary is obvious by Theorem~\ref{thm:priori-decay-b_0}. 

\begin{corollary}\label{cor:19020718}
There is no pure point spectrum for $H$, that is  $\sigma_{\mathrm{pp}}(H)=\emptyset$.
\end{corollary}
\begin{remark*}
Theorem~\ref{thm:priori-decay-b_0} and hence Corollary~\ref{cor:19020718}
hold true also for an escape function  
\begin{align} 
f_1=f_1(x,y) 
=\chi(x)+[1-\chi(x)]|x|^{1/2} 
\label{eq:190323} 
\end{align} 
instead of \eqref{eq:181215}. Obviously
Theorem~\ref{thm:priori-decay-b_0} is a  stronger statement with $f$
rather than  $f$ replaced by $f_1$. 
The setting with  $f_1$ is very similar to the one of \cite{Y1}. 
We note that \cite{Y1} does not discuss absence of eigenvalues,
and that the assumptions are not completely comparable. 
For example we allow a growing long-range part in the direction of  the Stark field,
while in \cite{Y1} the potential can only grow in the classically
forbidden region. In the direction of the field the potential  in
\cite{Y1} is assumed to be short-range. On the other hand
the  singular part in \cite{Y1} 
can have unbounded support. 
\end{remark*}

Our second theorem is \emph{LAP bounds} for the resolvent
\begin{equation*}
R(z)=(H-z)^{-1}\ \ \text{for }z\in\mathbb C\setminus\mathbb R.
\end{equation*}
We shall need  an additional
condition to treat the classical forbidden region.
\begin{cond}\label{cond:1803131348b}
Conditions~\ref{cond:1803131348} and \ref{cond:smooth2wea3n2} hold.
In addition, for any $f_0\ge 1$
\begin{align*}
\lim_{\mu\to-\infty}\Bigl(\inf\bigl\{-x+q(x,y)\,|\,x<\mu,\, f(x,y)\le f_0\bigr\}\Bigr)
=\infty
.
\end{align*}
\end{cond}

For a compact interval $I \subseteq \R$ we write 
\begin{equation*}
I_\pm 
= \bigl\{z\in\C \, \big|\, \mathop{\mathrm{Re}}z\in I, \ \pm\mathop{\mathrm{Im}}z\in (0,1)\bigr\},
\end{equation*}
respectively. 
In addition, we introduce a differential operator $p^f$ and a matrix $\ell$ as 
\begin{align}
p^f=-\mathrm i\partial^f,\quad 
\ell_{jk} 
&=|\partial f|^2\delta_{jk}-(\partial_j f)(\partial_k f),
\label{eq:180925b}
\end{align}
cf.\ \eqref{eq:180925}.
Note that $\ell$ represents a projection onto the orthogonal complement of 
$\mathop{\mathrm{grad}}f$, scaled by $|\partial f|^2$.
In particular, $\ell$ is non-negative. 

The {Einstein summation convention} is adopted throughout the paper,
although tensorial superscripts are avoided.
For a general linear operator $T$
we write $\langle T \rangle_\psi = \langle \psi, T\psi \rangle$.

\begin{thm}\label{thm:lap-bounds}
Assume Condition~\ref{cond:1803131348b}. Let $I\subseteq\R$ be a compact interval.
Then there exists $C>0$ such that for any $z\in I_\pm$ and $\psi\in\vB$ 
the state $\phi=R(z)\psi$ satisfies
\begin{equation*}
\|\phi\|_{\vB^*} 
+ \|(1-x/r)^{1/2}\phi\|_{L^2_{-1/2}}
+ \|p^f\phi\|_{\vB^*} 
+ \langle p_jf^{-1}\ell_{jk}p_k \rangle_{\phi} ^{1/2}
\le C\|\psi\|_\vB.
\end{equation*}
\end{thm}
\begin{remark*}

The finiteness of the second term on the left-hand side means that 
$\phi$ has a slightly stronger decay rate 
in directions not parallel to $x$, cf.\ \eqref{eq:inclusion-relations}. 
The bound  actually reproduces a result of \cite{A} for the $1$-body case.
Similarly, the derivatives $p\phi$ have slightly stronger decay rates
in directions orthogonal to $\mathop{\mathrm{grad}}f$, as expressed by
the finiteness of the fourth term,
 cf. \eqref{eq:18121514} below. 
\end{remark*}

The  proof of Theorem~\ref{thm:lap-bounds} will be given in Section~\ref{sec:181010}.
The following corollary follows directly from Theorem~\ref{thm:lap-bounds}. 

\begin{corollary}\label{cor:19020719}
There is no singularly continuous spectrum for $H$, that is  $\sigma_{\mathrm{sc}}(H)=\emptyset$.
\end{corollary}

\begin{remark*}
Corollaries~\ref{cor:19020718} and \ref{cor:19020719} 
assert that  the spectrum $\sigma(H)$ is purely absolutely continuous. 
Although Theorems ~\ref{thm:priori-decay-b_0} and \ref{thm:lap-bounds} are much more detailed results,
stability of purely absolute continuous spectrum is of its own interest.
See e.g.\ \cite{BCDSSW,NP,K,Sa,CK} for related results, most of which depend
on $1$-dimensional techniques.
\end{remark*}

Thirdly, we provide  \emph{radiation condition bounds} for $R(z)$,
which describe the leading oscillation of the resolvent along $\mathop{\mathrm{grad}} f$.
Define a differential operator $A$ as 
\begin{equation} \label{eq:def-A}
A=[H_0,\mathrm if]=\mathop{\mathrm{Re}}p^f
=p^f-\tfrac{\mathrm i}2(\Delta f),
\end{equation}
cf.\ \eqref{eq:180925b}. Note that $A$ is essentially self-adjoint on
$C^\infty_\c(\R^d)$, and by  using
Condition \ref{cond:1803131348} and  Remark \ref{remark:fbndg} one
easily checks that $
\vD(A)\supseteq \vD(H)$. Let
$I\subseteq \mathbb R$ be a compact interval, and we choose an
asymptotic complex phase $a=a_z$ as
\begin{equation}\label{eq:phase-a}
a= \bar\chi_l
\Bigl[\sqrt{(r-q_1+z)/r} \pm \tfrac{\mathrm i}4fr^{-2}\Bigr]
\ \ \text{for }z\in I_\pm,
\end{equation}
respectively.
Here $l\in\mathbb N_0$ is chosen dependently on $I$ such that 
$\mathop{\mathrm{Re}}\bigl((r-q_1+z)/r\bigr)$ is uniformly positive 
for all $z\in I_\pm$ and $(x,y)\in\mathop{\mathrm{supp}}\bar\chi_l$. 
The branch of square root is fixed such that 
$\mathop{\mathrm{Re}}\sqrt w>0$ for $w\in\mathbb C\setminus (-\infty,0]$.

Let us further impose an additional condition 
that slightly strengthens the second bound of \eqref{eq:1812151221}. 
Let us use shorthand notation 
\begin{align}
\widetilde \partial=|\partial f|\partial,\quad 
\widetilde p=|\partial f|p.
\label{eq:18121512}
\end{align}
Note that then in particular  we have 
\begin{align}
(p^f)^*p^f+p_j\ell_{jk}p_k=p|\partial f|^2p=(\widetilde p)^*\widetilde p.
\label{eq:18121514}
\end{align}

\begin{cond}\label{cond:additional-condition}
Condition~\ref{cond:1803131348b} holds.
In addition, 
there exist $\tilde\rho, C>0$ such that
\begin{equation*}
|\widetilde\partial q_1|\le Cf^{-1-\tilde\rho}.
\end{equation*}
\end{cond}

With $\rho,\tilde\rho>0$ from Condition~\ref{cond:1803131348} and \ref{cond:additional-condition} we set 
\begin{equation*}
\beta_c = \min\{1/2,\rho,\tilde\rho\}.
\end{equation*}

\begin{thm}\label{thm:RC-bound}
Assume Condition~\ref{cond:additional-condition}.
Let $I\subseteq\R$ be a compact interval, 
and choose $l\in\mathbb N_0$ as above.
Then for all 
$\beta\in [0,\beta_c)$ there exists $C>0$ such that 
for any $z\in I_\pm$ and $\psi\in f^{-\beta}\vB$ 
the states $\phi=R(z)\psi$ satisfy
\begin{equation*}
\|f^\beta(1-x/r)^{1/2}\phi\|_{L^2_{-1/2}} 
+\|f^\beta(A\mp a)\phi\|_{\vB^*} 
+ \langle p_jf^{2\beta-1}\ell_{jk}p_k \rangle_{\phi}^{1/2} 
\le 
C\|f^\beta\psi\|_\vB,
\end{equation*}
respectively.
\end{thm}
\begin{remark}\label{rem:1812248}
Our choice   of $a$ is partly taken for technical convenience. It is not
claimed to be canonical and we do not 
 consider Theorem \ref{cond:additional-condition} to be  optimal for
$\rho,\tilde \rho>1/2$. In fact we show in \cite{AIIS2} that in some cases,
 some $\beta>1/2$ are  allowed for a  different choice  of $a$ still having
$\|f^\beta(A\mp  a)\phi\|_{\vB^*}$ 
 finite for $\psi$ in a `good
space'. We take below (see Corollary \ref{cor:RC-bound-real}) the
spectral parameter $z$ to the real axis. Considering  for simplicity
only $z=0$ and
$q=0$ indeed $a=a_\pm^{{\rm sim}}:= \tfrac {f^2}{2r}\pm \tfrac{\mathrm i}4fr^{-2}$
 is a better choice in the sense that in fact any $\beta\in [0,4)$ can be
chosen in that case (note that intuitively $\tfrac {f^2}{2r}\approx
1$).  Note for comparison that Corollary \ref{cor:RC-bound-real} in
this case implies the bounds
\begin{align*}
 \|f^\beta(A\mp a_\pm^{{\rm sim}})\phi\|_{\vB^*} 
\le 
C_\beta\|f^\beta\psi\|_\vB\text{ for }\beta<1/2, 
\end{align*} but the result does not  imply this  bound if  $\beta\geq 1/2$.

In
\cite{AIIS2} we construct   WKB
approximations. For the above simple case these read
\begin{align*}
  \bigl(f^{d-2}r\bigr)^{-1/2}\exp\bigl(\pm\mathrm i\tfrac13f^3\bigr)\xi(y/f)\in\vB^*\setminus\vB_0^*
\end{align*} for a dense set of functions $\xi\in L^2(\R^{d-1})$ in the  other parabolic
coordinates
 $g=y/f$. Radiation bounds are  related to   WKB
approximations. Thus manifestly 
\begin{align*}
  (A\mp a_\pm^{{\rm sim}})\bigl(f^{d-2}r\bigr)^{-1/2}\exp\bigl(\pm\mathrm
  i\tfrac13f^3\bigr)\xi(g) =0\text{ for }f>\sqrt 2.
\end{align*} Of course this assertion relies on the particular form of
$ a_\pm^{{\rm sim}}$ (including the particular 
imaginary part).

\end{remark}

A proof of Theorem~\ref{thm:RC-bound} will be given in Section~\ref{sec:181207}.

Finally we present applications of Theorems~\ref{thm:priori-decay-b_0}, 
\ref{thm:lap-bounds} and \ref{thm:RC-bound}. 
The first application is LAP. (We distinguish between `LAP bounds' and `LAP'.)

\begin{corollary}\label{cor:Limiting-Absorption-Principle-Stark}
Assume Condition~\ref{cond:additional-condition}.
Let $I\subseteq\R$ be a compact interval.
For any $k=0,1$, $s>1/2$ and $\epsilon\in(0,\min\set{\beta_c,s-1/2})$ 
there exists $C>0$ such that for any $z,z'\in I_+$ or $z,z'\in I_-$
\begin{equation}\label{eq:Holder-continuity}
\begin{split}
\|\widetilde p^kR(z)-\widetilde p^kR(z')\|_{\vL(L^2_s,L^2_{-s})}\le C|z-z'|^\epsilon.
\end{split}
\end{equation}
In particular, $\widetilde p^kR(z)$ for $k=0,1$ have uniform limits 
as $I_\pm\ni z\to\lambda\in I$ in the norm topology of $\vL(L^2_s,L^2_{-s})$, 
which one denotes by 
\begin{equation}\label{eq:uniform-limit-z-to-lambda}
\begin{split}
\widetilde p^kR(\lambda\pm \mathrm i0) = \lim_{z\to \lambda\pm \mathrm i0}\widetilde p^kR(z)
\ \ \text{in }\vL(L^2_s,L^2_{-s}), 
\end{split}
\end{equation}
respectively.
Moreover, these limits $\widetilde p^kR(\lambda\pm \mathrm i0)$ belong to $\vL(\vB,\vB^*)$.
\end{corollary}

Combining Theorem~\ref{thm:RC-bound} 
and Corollary~\ref{cor:Limiting-Absorption-Principle-Stark} 
we obtain  radiation condition bounds for real spectral parameters by
taking limits. 
Thus we need  respective limits
$$a_\pm:=\lim_{z\to\lambda\pm\mathrm i0}a_z\ \ \text{for } \lambda\in I.$$

\begin{corollary}\label{cor:RC-bound-real}
Assume Condition~\ref{cond:additional-condition}. 
Let $I\subseteq\R$ be a compact interval,
and choose $l\in\mathbb N_0$ as above.
Then for all $\beta\in [0,\beta_c)$ there exists $C>0$ such that 
for any $\lambda\in I$ and $\psi\in f^{-\beta}\vB$ 
the states $\phi=R(\lambda\pm\mathrm i0)\psi$ satisfy
\begin{equation*}
\|f^\beta(1-x/r)^{1/2}\phi\|_{L^2_{-1/2}} 
+\|f^\beta(A\mp a_\pm)\phi\|_{\vB^*} 
+ \langle p_jf^{2\beta-1}\ell_{jk}p_k \rangle_{\phi}^{1/2} 
\le 
C\|f^\beta\psi\|_\vB,
\end{equation*}
respectively.
\end{corollary}

As the last application, we provide \emph{Sommerfeld's uniqueness result}.

\begin{corollary}\label{cor:Sommerfeld-unique-result}
Assume Condition~\ref{cond:additional-condition}.
Let $\lambda\in\R$, 
$\phi\in f^{\beta}\vB^*$ and $\psi\in f^{-\beta}\vB$ with $\beta\in[0,\beta_c)$.
Then $\phi=R(\lambda\pm \mathrm i0)\psi$ holds if and only if both of the following hold:
\begin{enumerate}[(1)]
\item\label{item:18122818}
$(H-\lambda)\phi=\psi$ in the distributional sense;
\item\label{item:18122819}
$(A\mp a_\pm)\phi\in f^{-\beta}\vB_0^*$.
\end{enumerate}
\end{corollary}

The proofs of Corollaries~\ref{cor:Limiting-Absorption-Principle-Stark}, \ref{cor:RC-bound-real} 
and \ref{cor:Sommerfeld-unique-result} will be given in Section~\ref{sec:181207}.

\section{Conjugate operator}\label{sec:Conjugate operator}

This is a short preliminary section for the proofs of our main theorems
in the following sections. 
Here we compute an explicit expression for a {weighted commutator} 
\begin{equation*}
[H, \mathrm iA]_{\Theta} :=\mathrm i( H\Theta A - A\Theta H).
\end{equation*}
 For various choices  of the \emph{weight function}
 $\Theta\in C^\infty(\mathbb R^d)$ (see 
 \eqref{eq:15.2.15.5.8bb}, \eqref{eq:15.2.15.5.8bbb}, 
\eqref{eq:1810021842} and \eqref{eq:theta-rc} for concrete
expressions)  this `commutator', with  $A$  given  as in
\eqref{eq:def-A}, tends to  be
 positive (for this reason $A$
 is referred to as a  \emph{conjugate operator}).
Implementation of a commutator is always haunted by the `domain problem',
however, as long as there is a common  core for operators involved, 
in the present case  $C^\infty_{\mathrm c}(\mathbb R^d)$, an approximation argument works easily.
In this paper we do not elaborate further on domains for readability. 
Actually we have rigorously treated such problems 
in  previous works in more complicated situations (like in cases with boundaries), cf. \cite{IS}.

For the moment we only assume that $\Theta\in C^\infty(\mathbb R^d)$ 
is a function only of $f$, and that for some $m\in\mathbb N_0$
and for all $k\in\mathbb N_0$
\begin{equation}\label{eq:theta-derivative}
\mathop{\mathrm{supp}}\Theta\subseteq\bigl\{(x,y)\in\mathbb R^d\,\big|\,f(x,y)\ge 2^m\bigr\}
,\quad 
|\Theta^{(k)}| \le C_k
,
\end{equation}
where $\Theta^{(k)}$ denotes the $k$-th derivative of $\Theta$ in $f$. 
In the later arguments we may let $m\in\mathbb N_0$ be sufficiently large, so that 
\begin{align}
\mathop{\mathrm{supp}}\Theta\cap \mathop{\mathrm{supp}}q_3=\emptyset.
\label{eq:190315}
\end{align}
We note that on $\mathop{\mathrm{supp}}\Theta$ we can write derivatives of $f$ as 
\begin{align}
\begin{split}
\partial_xf
&=
\tfrac12 fr^{-1}
,\\
\partial_yf&
=\tfrac12f^{-1}r^{-1}y, \\
\partial_x\partial_xf
&=
\tfrac12 f^{-1}r^{-1}
-\tfrac14 fr^{-2}
-\tfrac12 f^{-1}r^{-3}x^2
,\\
\partial_{y_j} \partial_{y_k}f
&
=
\tfrac12f^{-1}r^{-1}\delta_{jk}
-\tfrac14f^{-3}r^{-2}y_jy_k
-\tfrac12f^{-1}r^{-3}y_jy_k
,\\
\partial_x\partial_yf
&=
\partial_y\partial_xf
=
-\tfrac14f^{-1}r^{-2}y
-\tfrac12 f^{-1}r^{-3}xy
.
\end{split}
\label{180814}
\end{align}
In particular, we also have expressions on $\mathop{\mathrm{supp}}\Theta$:
\begin{align}
\begin{split}
\partial^f
&=\tfrac12fr^{-1}\partial_x+\tfrac12f^{-1}r^{-1} y\partial_y,\\
(\partial^f r)&=\tfrac12fr^{-1} ,\\
|\partial f|^2&=
\tfrac12 r^{-1}
,\\ 
\Delta f
&
=\tfrac12(d-2)f^{-1}r^{-1}
,\\
\partial_j\partial_kf
&=
f^{-1}\ell_{jk}
-f^{-1}|\partial f|^2 (\partial_jr)(\partial_kr),
\end{split}
\label{eq:180816}
\end{align}
see \eqref{eq:180925} and \eqref{eq:180925b} for $\partial^f$ and $\ell$, respectively.

\begin{lemma}\label{lemma:commutator-calculation}
Assume Condition~\ref{cond:1803131348}. 
Then, as quadratic forms on $C_{\mathrm c}^{\infty}({\R}^d)$,
\begin{align*}
[H,\mathrm iA]_\Theta
&
=
A\Theta' A
+p_jf^{-1}\Theta \ell_{jk}p_k
+p_jf^{-1}|\partial f|^2\bigl(\delta_{jk}-(\partial_jr)(\partial_kr)\bigr)\Theta p_k
\\&\phantom{{}={}}{}
+\tfrac12 f^{-1}(1-x/r)\Theta 
-\tfrac14|\partial f|^4\Theta'''
-\tfrac12(\partial^f|\partial f|^2)\Theta''
-\tfrac12f^{-1}|\partial f|^4\Theta'' 
\\&\phantom{{}={}}{}
+ q_4\Theta'
+ q_5\Theta 
-2\mathop{\mathrm{Im}}\bigl(q_2\Theta p^f\bigr)
-2\mathop{\mathrm{Re}}\bigl(f^{-1}|\partial f|^2\Theta H\bigr)
-\mathop{\mathrm{Re}}\bigl(|\partial f|^2\Theta' H\bigr)
\end{align*}
with 
\begin{align*}
 q_4
&=
-\tfrac14(\Delta|\partial f|^2)
-\tfrac12f^{-1}|\partial f|^2(\Delta f) 
-f^{-1}(\partial^f |\partial f|^2)
+f^{-2}|\partial f|^4
+|\partial f|^2q_2 
,\\
 q_5
&=
-\tfrac14(\Delta^2f) 
-\tfrac12f^{-1}(\Delta |\partial f|^2) 
+\tfrac12f^{-2}|\partial f|^2(\Delta f) 
+f^{-2}(\partial^f|\partial f|^2) 
-f^{-3}|\partial f|^4 
\\&\phantom{{}={}}{}
+2f^{-1}|\partial f|^2q  
-(\partial^f q_1) 
+(\Delta f)q_2 
.
\end{align*}
In particular, 
\begin{align*}
| q_4|
\le Cf^{-1-\min\{3,\rho\}}r^{-1}
,\quad 
 q_5
\ge -Cf^{-1-\min\{6,\rho\}}
.
\end{align*}
\end{lemma}
\begin{proof}
Using the adjoint of the expression \eqref{eq:def-A}, we can compute 
\begin{align}
\begin{split}
[H,\mathrm iA]_\Theta
&
=
\mathop{\mathrm{Im}}\bigl((p^f)^*\Theta p^2\bigr)
+2\mathop{\mathrm{Im}}\bigl((p^f)^*\Theta (-x+q)\bigr)
+\mathop{\mathrm{Re}}\bigl((\Delta f)\Theta H\bigr)
\\&
=
(p^f)^*\Theta' p^f
+p_j(\partial_j\partial_kf)\Theta p_k
-\tfrac12p(\Delta f)\Theta p
-\tfrac12p|\partial f|^2\Theta' p
\\&\phantom{{}={}}{}
+(x-q_1)(\Delta f)\Theta 
+(x-q_1)|\partial f|^2 \Theta'
+(\partial_xf)\Theta 
-(\partial^f q_1)\Theta 
\\&\phantom{{}={}}{}
-2\mathop{\mathrm{Im}}\bigl(q_2\Theta p^f\bigr)
+\mathop{\mathrm{Re}}\bigl((\Delta f)\Theta H\bigr).
\end{split}
\label{eq:180628}
\end{align}
We combine the third, fifth and tenth terms of \eqref{eq:180628} as 
\begin{align}
\begin{split}
&
-\tfrac12p(\Delta f)\Theta p
+(x-q_1)(\Delta f)\Theta 
+\mathop{\mathrm{Re}}\bigl((\Delta f)\Theta H\bigr)
\\&
=
-\tfrac14\bigl(\Delta(\Delta f)\Theta \bigr)
+(\Delta f)q_2\Theta 
,
\end{split}
\label{eq:180628c}
\end{align}
and, similarly, the fourth and sixth terms of \eqref{eq:180628} as 
\begin{align}
\begin{split}
&
-\tfrac12p|\partial f|^2\Theta' p
+(x-q_1)|\partial f|^2 \Theta'
\\&
=
-\tfrac14\bigl(\Delta|\partial f|^2\Theta'\bigr)
+|\partial f|^2q_2 \Theta'
-\mathop{\mathrm{Re}}\bigl(|\partial f|^2\Theta' H\bigr).
\end{split}
\label{eq:180628d}
\end{align}
In addition, let us add to the right-hand side of \eqref{eq:180628} the following ``zero'' term:
\begin{align}
\begin{split}
0
&=
pf^{-1}|\partial f|^2\Theta p
-2xf^{-1}|\partial f|^2\Theta 
-\tfrac12\bigl(\Delta f^{-1}|\partial f|^2\Theta\bigr)
\\&\phantom{{}={}}{}
+2f^{-1}|\partial f|^2q\Theta 
-2\mathop{\mathrm{Re}}\bigl(f^{-1}|\partial f|^2\Theta H\bigr)
\end{split}
\label{eq:180628e}
\end{align}
Then by \eqref{eq:180628}, \eqref{eq:180628c}, \eqref{eq:180628d} and \eqref{eq:180628e}
we obtain 
\begin{align}
\begin{split}
[H,\mathrm iA]_\Theta
&
=
(p^f)^*\Theta' p^f
+pf^{-1}|\partial f|^2\Theta p
+p_j(\partial_j\partial_kf)\Theta p_k
+(\partial_xf)\Theta 
\\&\phantom{{}={}}{}
-2xf^{-1}|\partial f|^2\Theta 
-\tfrac14\bigl(\Delta(\Delta f)\Theta \bigr)
-\tfrac14\bigl(\Delta|\partial f|^2\Theta'\bigr)
-\tfrac12\bigl(\Delta f^{-1}|\partial f|^2\Theta\bigr)
\\&\phantom{{}={}}{}
+|\partial f|^2q_2 \Theta'
+2f^{-1}|\partial f|^2q\Theta 
-(\partial^f q_1)\Theta 
+(\Delta f)q_2\Theta 
\\&\phantom{{}={}}{}
-2\mathop{\mathrm{Im}}\bigl(q_2\Theta p^f\bigr)
-2\mathop{\mathrm{Re}}\bigl(f^{-1}|\partial f|^2\Theta H\bigr)
-\mathop{\mathrm{Re}}\bigl(|\partial f|^2\Theta' H\bigr)
.
\end{split}
\label{1808145}
\end{align}

Next, we expand the sixth to eighth terms of \eqref{1808145} as 
\begin{align}
\begin{split}
&
-\tfrac14\bigl(\Delta(\Delta f)\Theta \bigr)
-\tfrac14\bigl(\Delta|\partial f|^2\Theta'\bigr)
-\tfrac12\bigl(\Delta f^{-1}|\partial f|^2\Theta\bigr)
\\&
=
-\tfrac14|\partial f|^4\Theta'''
-\tfrac12|\partial f|^2(\Delta f)\Theta''
-\tfrac12(\partial^f|\partial f|^2)\Theta''
-\tfrac12f^{-1}|\partial f|^4\Theta''
\\&\phantom{{}={}}{}
-\tfrac12(\partial^f\Delta f)\Theta'
-\tfrac14(\Delta|\partial f|^2)\Theta'
-\tfrac14(\Delta f)^2\Theta'
-\tfrac12f^{-1}|\partial f|^2(\Delta f)\Theta'
\\&\phantom{{}={}}{}
-f^{-1}(\partial^f|\partial f|^2)\Theta'
+f^{-2}|\partial f|^4\Theta'
-\tfrac14(\Delta^2f)\Theta 
-\tfrac12f^{-1} (\Delta |\partial f|^2)\Theta
\\&\phantom{{}={}}{}
+\tfrac12f^{-2}|\partial f|^2(\Delta f)\Theta
+f^{-2}(\partial^f|\partial f|^2)\Theta
-f^{-3}|\partial f|^4\Theta
.
\end{split}
\label{180925}
\end{align}
Then the first term of \eqref{1808145}
combined with the second, fifth and seventh terms of \eqref{180925}
makes the first term of asserted identity. 
Inserting expressions from \eqref{180814} and \eqref{eq:180816}
into the second to fourth terms of \eqref{1808145},
we obtain the second to third terms of the asserted identity.
The rest terms of \eqref{1808145} and \eqref{180925}
clearly correspond to the rest terms of the asserted identity.
Hence we are done.
\end{proof}

\section{Rellich's theorem}\label{sec:180629}

In this section we prove Theorem~\ref{thm:priori-decay-b_0}. 
The proof reduces to the following two propositions. 
We basically proceed along the lines of \cite{IS}, but here
we need to discuss \emph{cubic} exponential decay estimates
while in \cite{IS} \emph{linear} exponential decay estimates suffice. 
This appears as a unique   feature for the Stark Hamiltonian.

Throughout the section we impose  Conditions~\ref{cond:1803131348} and
\ref{cond:smooth2wea3n2}.

\begin{proposition}\label{prop:Absence of super-exponentially decaying ef}
If a function $\phi\in L^2_{\mathrm{loc}}(\mathbb R^d)$ satisfies for
some $m_0\in\mathbb N_0$:
\begin{enumerate}
\item
$(H-\lambda)\phi=0$ in the distributional sense, 
\item
$\bar\chi_{m_0}\phi\in \mathcal B^*_0$,
\end{enumerate}
then $\bar\chi_{m_0}\exp(\alpha f^3)\phi\in \mathcal B_0^*$ for any $\alpha\ge 0$.
\end{proposition}

\begin{proposition}\label{prop:Absence of super-exponentially decaying efb}
If a function $\phi\in L^2_{\mathrm{loc}}(\mathbb R^d)$ satisfies for
some $m_0\in\mathbb N_0$:
\begin{enumerate}
\item
$(H-\lambda)\phi=0$ in the distributional sense, 
\item
$\bar\chi_{m_0}\exp(\alpha f^3)\phi\in\mathcal B_0^*$ for any $\alpha\ge 0$,
\end{enumerate}
then $\phi=0$ on $\mathbb R^d$.
\end{proposition}

Propositions~\ref{prop:Absence of super-exponentially decaying ef} 
and \ref{prop:Absence of super-exponentially decaying efb} will be proved 
in Subsections~\ref{subsec:180629} and \ref{subsec:18062923}, respectively.

\subsection{A priori super-cubic-exponential decay estimate}\label{subsec:180629}

Here we are going to 
prove Proposition~\ref{prop:Absence of super-exponentially decaying ef}. 
Choose a weight function $\Theta$ to be of the form
\begin{align}
\begin{split}
\Theta&= \Theta_{m,n,\nu}^{\alpha,\beta,\delta}
=\chi_{m,n}\mathrm e^\theta;\\
\theta&=\theta_\nu^{\alpha,\beta,\delta}
=2\alpha f^3+6\beta\int_0^fs^{2}(1+s/2^\nu)^{-3-\delta}\,\mathrm ds
\end{split}
\label{eq:15.2.15.5.8bb}
\end{align}
with parameters $\alpha,\beta\ge 0,\ \delta>0$ and $m,n,\nu\in\mathbb N$.
Note that $\Theta$ actually satisfies \eqref{eq:theta-derivative}
for large $m$.
We denote the derivatives in $f$ by primes as before.
If we set 
\begin{align*}
\theta_0=1+f/2^\nu
\end{align*} 
for notational simplicity, then 
\begin{align*}
\begin{split}
\theta'=6\alpha f^2+6\beta f^2\theta_0^{-3-\delta},\quad
\theta''=12\alpha f+12\beta f\theta_0^{-3-\delta}
-6\beta (3+\delta)2^{-\nu}f^2\theta_0^{-4-\delta}, \ldots.
\end{split}
\end{align*}
Noting that $2^{-\nu}\theta_0^{-1}\leq f^{-1}$, we have
\begin{align*}
\bigl|(\theta-2\alpha f^3)^{(k)}\bigr|\leq C_{\delta,k} \beta f^{3-k}\theta_0^{-3-\delta}
\quad \text{for }k=1,2,\dots.
\end{align*}

The following form inequality plays an essential role in the proof of 
Proposition~\ref{prop:Absence of super-exponentially decaying ef}.
\begin{lemma}\label{lem:14.10.4.1.17ffaabb}
Fix any $\alpha_0\ge 0$ and $\delta\in(0,\min\{2,\rho\})$ 
in the definition \eqref{eq:15.2.15.5.8bb} of $\Theta$.
Then there exist $\beta,c,C>0$ and $n_0\in\mathbb N$ 
such that uniformly in 
$\alpha\in[0,\alpha_0]$, $n>m\ge n_0$ and $\nu\ge n_0$, 
as quadratic forms on $\mathcal D(H)$,
\begin{align}
\begin{split}
\mathop{\mathrm{Im}}\bigl(A\Theta (H-\lambda)\bigr)
&
\ge 
c f^{-1}\theta_0^{-\delta}\Theta 
-Cf^{-1}\bigl(\chi_{m-1,m+1}^2+\chi_{n-1,n+1}^2\bigr)\mathrm e^\theta
\\&\phantom{{}={}}{}
+\mathop{\mathrm{Re}}\bigl(\gamma (H-\lambda)\bigr)
,
\end{split}
\label{eq:14.9.26.9.53ffaabb}
\end{align}
where $\gamma=\gamma_{m,n,\nu}$ is a certain function
satisfying $\mathop{\mathrm{supp}}\gamma\subseteq\mathop{\mathrm{supp}}\chi_{m,n}$ 
and $|\gamma|\le C\mathrm e^\theta$
uniformly in $n>m\ge n_0$ and $\nu\ge n_0$.
\end{lemma}

\begin{proof}
Fix any $\alpha_0\ge 0$ and $\delta\in(0,\min\{2,\rho\})$. 
We will fix small $\beta\in(0,1]$ and large $n_0\in\mathbb N$ 
in the last step of the proof.
For the moment we only assume $n_0\in\mathbb N$ is large enough that 
\eqref{eq:190315} holds for all $m\ge n_0$.
All the estimates below 
are uniform in $\alpha\in[0,\alpha_0]$, $\beta\in (0,1]$, $n>m\ge n_0>0$ and $\nu\ge n_0$.

By Lemma~\ref{lemma:commutator-calculation} we can bound  
\begin{align}
\begin{split}
&2\mathop{\mathrm{Im}}\bigl(A\Theta (H-\lambda)\bigr)
\\&=
[H,\mathrm iA]_\Theta
+\lambda|\partial f|^{2}\Theta'
\\&
\ge 
A\theta'\Theta A
+p_jf^{-1}\Theta\ell_{jk}p_k
+\tfrac12 f^{-1}(1-x/r)\Theta 
\\&\phantom{{}={}}{}
-\tfrac34|\partial f|^4\theta'\theta''\Theta
-\tfrac14|\partial f|^4\theta'^3\Theta
-\tfrac12(\partial^f|\partial f|^2)\theta'^2\Theta
-\tfrac12f^{-1}|\partial f|^4\theta'^2\Theta
\\&\phantom{{}={}}{}
-C_1Q
-2\mathop{\mathrm{Re}}\bigl(f^{-1}|\partial f|^2\Theta (H-\lambda)\bigr)
-\mathop{\mathrm{Re}}\bigl(|\partial f|^2\Theta' (H-\lambda)\bigr)
\\&
\ge 
\bigl(A+\tfrac{\mathrm i}2|\partial f|^2\theta'\bigr)\theta'\Theta \bigl(A-\tfrac{\mathrm i}2|\partial f|^2\theta'\bigr)
+p_jf^{-1}\Theta\ell_{jk}p_k
+\tfrac12 f^{-1}r^{-1}(r-x)\theta_0^{-\delta}\Theta 
\\&\phantom{{}={}}{}
+\tfrac14|\partial f|^4\theta'\bigl(\theta''-2f^{-1}\theta'\bigr)\Theta
-C_2Q
+\mathop{\mathrm{Re}}\bigl(\gamma_1(H-\lambda)\bigr)
.
\end{split}
\label{eq:18062820c}
\end{align}
Here and below we gather \emph{admissible} error terms in $Q$,
which is of the form
\begin{align*}
Q
&=
f^{-4}|\chi_{m,n}'''|\mathrm e^\theta
+f^{-2}|\chi_{m,n}''|\mathrm e^\theta
+|\chi_{m,n}'|\mathrm e^\theta
+f^{-1-\min\{2,\rho\}}\Theta
\\&\phantom{{}={}}{}
+pr^{-1}|\chi_{m,n}'|\mathrm e^\theta p
+pf^{-1-\rho}r^{-1}\Theta p
.
\end{align*}
Actually $-Q$ can be bounded below as 
\begin{align}
\begin{split}
-Q
&\ge 
-C_3f^{-1-\min\{2,\rho\}}\Theta
-C_3f^{-1}\bigl(\chi_{m-1,m+1}^2+\chi_{n-1,n+1}^2\bigr)\mathrm e^\theta
\\&\phantom{{}={}}{}
-2\mathop{\mathrm{Re}}\bigl(r^{-1}|\chi_{m,n}'|\mathrm e^\theta  (H-\lambda)\bigr)
-2\mathop{\mathrm{Re}}\bigl(f^{-1-\rho}r^{-1}\Theta(H-\lambda)\bigr)
,
\end{split}
\label{eq:180630c}
\end{align}
and we will see that this will be negligibly absorbed 
by other terms of \eqref{eq:18062820c}. 

Let us further combine and bound the first and second terms of \eqref{eq:18062820c} in the following manner.
Choose $n_0=n_{0,\beta}$ large enough depending on $\beta\in(0,1]$,
so that 
$$
\theta'\ge 
6\beta (f\theta_0^{-1})^3f^{-1}\theta_0^{-\delta}
\ge 
6\beta 2^{3(n_0-1)}f^{-1}\theta_0^{-\delta}
\ge \tfrac12f^{-1}\theta_0^{-\delta}\ \ 
\text{on }\mathop{\mathrm{supp}}\Theta.$$
Then we have the first and second terms of \eqref{eq:18062820c} bounded as 
\begin{align}
\begin{split}
&
\bigl(A+\tfrac{\mathrm i}2|\partial f|^2\theta'\bigr)\theta'\Theta \bigl(A-\tfrac{\mathrm i}2|\partial f|^2\theta'\bigr)
+p_jf^{-1}\Theta\ell_{jk}p_k
\\&
\ge 
\tfrac12\bigl(A+\tfrac{\mathrm i}2|\partial f|^2\theta'\bigr)
f^{-1}\theta_0^{-\delta}\Theta \bigl(A-\tfrac{\mathrm i}2|\partial f|^2\theta'\bigr)
+\tfrac12p_jf^{-1}\theta_0^{-\delta}\Theta\ell_{jk} p_k
\\&
\ge 
\tfrac12(p^f)^*f^{-1}\theta_0^{-\delta}\Theta p^f
-\tfrac18f^{-1}|\partial f|^4\theta'^2\theta_0^{-\delta}\Theta 
+\tfrac12p_jf^{-1}\theta_0^{-\delta}\Theta\ell_{jk} p_k
-C_4Q
\\&
\ge 
\tfrac12pf^{-1}|\partial f|^2\theta_0^{-\delta}\Theta p
-\tfrac18f^{-1}|\partial f|^4\theta'^2\theta_0^{-\delta}\Theta 
-C_4Q
\\&
\ge 
f^{-1}|\partial f|^2x\theta_0^{-\delta}\Theta 
+\tfrac18f^{-1}|\partial f|^4\theta'^2\theta_0^{-\delta}\Theta 
-C_5Q
+\mathop{\mathrm{Re}}\bigl(f^{-1}|\partial f|^2\theta_0^{-\delta}\Theta (H-\lambda)\bigr).
\end{split}
\label{eq:1806302c}
\end{align}
On the other hand, it is clear that the fourth term of \eqref{eq:18062820c}
is bounded as 
\begin{align}
\tfrac14|\partial f|^4\theta'\bigl(\theta''-2f^{-1}\theta'\bigr)\Theta
\ge 
-C_6\beta f^{-1}\theta_0^{-\delta}\Theta.
\label{eq:1806302d}
\end{align}
Now by \eqref{eq:18062820c}, \eqref{eq:1806302c}, \eqref{eq:1806302d} 
and \eqref{eq:180630c} we obtain 
\begin{align*}
\begin{split}
2\mathop{\mathrm{Im}}\bigl(A\Theta (H-\lambda)\bigr)
&\ge 
\tfrac12 (1-C_7\beta)f^{-1}\theta_0^{-\delta}\Theta 
-C_7f^{-1-\min\{2,\rho\}}\Theta
\\&\phantom{{}={}}{}
-C_7f^{-1}\bigl(\chi_{m-1,m+1}^2+\chi_{n-1,n+1}^2\bigr)\mathrm e^\theta
+\mathop{\mathrm{Re}}\bigl(\gamma_2(H-\lambda)\bigr)
.
\end{split}
\end{align*}
By choosing $\beta\in (0,1]$ small enough, 
and then $n_0\in\mathbb N$ large enough
we obtain the asserted inequality.
\end{proof}

\begin{proof}[Proof of Proposition~\ref{prop:Absence of super-exponentially decaying ef}]
Let $\phi\in  L^2_{\mathrm{loc}}(\mathbb R^d)$ and $m_0\in \N_0$
satisfy the assumptions of the assertion, and set
\begin{align*}
\alpha_0=\sup\bigl\{\alpha\ge 0\,\big|\,{\bar\chi_{m_0}}\exp(\alpha f^3)\phi\in \mathcal B_0^*\bigr\}.
\end{align*}
Assume $\alpha_0<\infty$, and let us deduce a contradiction.
Fix any $\delta\in (0,\min\{2,\rho\})$,
and choose $\beta>0$ and $n_0\ge 0$ as in Lemma~\ref{lem:14.10.4.1.17ffaabb}.
We may let $n_0\ge m_0+3$ without loss of generality.
If $\alpha_0=0$, let $\alpha=0$ so that we automatically have $\alpha+\beta>\alpha_0$.
Otherwise, we choose $\alpha\in[0,\alpha_0)$ such that $\alpha+\beta>\alpha_0$.
With such $\alpha$ and $\beta$
we evaluate the inequality \eqref{eq:14.9.26.9.53ffaabb}
in the state $\chi_{m-2,n+2}\phi$, 
and then we obtain for any $n>m\ge n_0$ and $\nu\ge n_0$
\begin{align}
\begin{split}
\bigl\|(f^{-1}\theta_0^{-\delta}\Theta)^{1/2}\phi\bigr\|^2
&
\le 
C_m\bigl\|\chi_{m-1,m+1}\phi\bigr\|^2
+C_\nu 2^{-n/2}\bigl\|\chi_{n-1,n+1}\exp(\alpha f^3)\phi\bigr\|^2
.
\end{split}
\label{eq:11.7.16.3.22a}
\end{align}
The second term on the right-hand side of (\ref{eq:11.7.16.3.22a})
  vanishes in the limit $n\to\infty$ since 
$\bar\chi_{m_0}\exp(\alpha f^3)\phi\in \mathcal B^*_0$, 
and hence by Lebesgue's monotone  convergence
theorem
\begin{align}
\bigl\|(\bar\chi_m f^{-1}\theta_0^{-\delta}
\mathrm{e}^{\theta})^{1/2}\phi\bigr\|^2
 &\le 
C_m\bigl\|\chi_{m-1,m+1}\phi\bigr\|^2.
\label{eq:11.7.16.3.43a}
\end{align}
Next we let $\nu \to\infty$ in \eqref{eq:11.7.16.3.43a}
    invoking again Lebesgue's monotone  convergence
theorem,
and then it follows that 
 $$\bar\chi_m^{1/2} f^{-1/2}\exp\bigl((\alpha+\beta)f^3\bigr)\phi
\in L^2(\mathbb R^d).$$ 
 This implies $\bar\chi_m^{1/2}\exp(\kappa f^3)\phi\in 
\mathcal B_0^*$
for any $\kappa\in(0,\alpha+\beta)$, which 
 contradicts that  $\alpha+\beta>\alpha_0$. 
\end{proof}

\subsection{Absence of super-cubic-exponentially decaying eigenfunctions}\label{subsec:18062923}

Next we prove Proposition~\ref{prop:Absence of super-exponentially decaying efb}. 
In order to prove it we choose $\Theta$ to be 
\begin{align}
\Theta= \Theta_{m,n}^{\alpha}
=\chi_{m,n}\mathrm e^\theta;\quad 
\theta=\theta^\alpha=2\alpha f^3
\label{eq:15.2.15.5.8bbb}
\end{align}
with parameters $\alpha\ge 1$ and $m,n\in\mathbb N$.
We first prove the following form inequality similar to Lemma~\ref{lem:14.10.4.1.17ffaabb},
however, focusing on different parameters.
We remark that  Lemma~\ref{lemma:alpha^2-estimate} will be implemented 
similarly to Lemma~\ref{lem:14.10.4.1.17ffaabb}.

\begin{lemma}\label{lemma:alpha^2-estimate}
There exist $c,C>0$ and $n_0\in\mathbb N$  
such that uniformly in 
$\alpha\ge 1$ and $n>m\ge n_0$, 
as quadratic forms on $\mathcal D(H)$,
\begin{align}
\begin{split}
\mathop{\mathrm{Im}}\bigl(A\Theta (H-\lambda)\bigr)
&
\ge 
c\alpha^2 f^3r^{-2}\Theta 
-C\alpha^2 f^{-1}\bigl(\chi_{m-1,m+1}^2+\chi_{n-1,n+1}^2\bigr)\mathrm e^\theta
\\&\phantom{{}={}}{}
+\mathop{\mathrm{Re}}\bigl(\gamma (H-\lambda)\bigr)
,
\end{split}
\label{eq:keyest}
\end{align}
where $\gamma=\gamma_{m,n}^\alpha$ is a certain function
satisfying $\mathop{\mathrm{supp}}\gamma\subseteq\mathop{\mathrm{supp}}\chi_{m,n}$ 
and $|\gamma|\le C\alpha\mathrm e^\theta$
uniformly in $\alpha\ge 1$ and $n>m\ge n_0$.
\end{lemma}

\begin{proof}
In this proof all the estimates are uniform in 
$\alpha\ge 1$ and $n>m\ge n_0$.
We will retake $n_0\in\mathbb N$ larger, if necessary, each time it appears below.

By Lemma~\ref{lemma:commutator-calculation} we bound 
\begin{align}
\begin{split}
&2\mathop{\mathrm{Im}}\bigl(A\Theta (H-\lambda)\bigr)
\\&
=
[H,\mathrm iA]_\Theta+\lambda|\partial f|^2\Theta'
\\&
\ge 
A\theta'\Theta A
+p_jf^{-1}\Theta\ell_{jk}p_k
+\tfrac12 f^{-1}(1-x/r)\Theta 
\\&\phantom{{}={}}{}
-\tfrac34|\partial f|^4\theta'\theta''\Theta
-\tfrac14|\partial f|^4\theta'^3\Theta
-\tfrac12(\partial^f|\partial f|^2)\theta'^2\Theta
-\tfrac12f^{-1}|\partial f|^4\theta'^2\Theta
\\&\phantom{{}={}}{}
-C_1Q
-2\mathop{\mathrm{Re}}\bigl(f^{-1}|\partial f|^2\Theta (H-\lambda)\bigr)
-\mathop{\mathrm{Re}}\bigl(|\partial f|^2\Theta' (H-\lambda)\bigr)
\\&
\ge 
\bigl(A+\tfrac{\mathrm i}2|\partial f|^2\theta'\bigr)\theta'\Theta \bigl(A-\tfrac{\mathrm i}2|\partial f|^2\theta'\bigr)
+p_jf^{-1}\Theta\ell_{jk} p_k
+\tfrac12 f^{-1}r^{-1}(r-x)\Theta 
\\&\phantom{{}={}}{}
-C_2Q
+\mathop{\mathrm{Re}}\bigl(\gamma_1(H-\lambda)\bigr)
,
\end{split}
\label{eq:18062820b}
\end{align}
where $Q$ consists of \emph{admissible} error terms:
\begin{align*}
Q
&=
f^{-4}|\chi_{m,n}'''|\mathrm e^\theta
+\alpha f^{-2}|\chi_{m,n}''|\mathrm e^\theta
+\alpha^2 |\chi_{m,n}'|\mathrm e^\theta
+\alpha f^{1-\min\{2,\rho\}}r^{-1}\Theta
\\&\phantom{{}={}}{}
+f^{-1-\min\{6,\rho\}}\Theta
+p r^{-1}|\chi_{m,n}'|\mathrm e^\theta p
+p f^{-1-\rho}r^{-1}\Theta p.
\end{align*}
Note that $Q$ satisfies 
\begin{align}
\begin{split}
-Q
&\ge 
-C_3\alpha^2 f^{3-\min\{2,\rho\}}r^{-2}\Theta
-C_3f^{-1-\min\{2,\rho\}}\Theta
\\&\phantom{{}={}}{}
-C_3\alpha^2 f^{-1}\bigl(\chi_{m-1,m+1}^2+\chi_{n-1,n+1}^2\bigr)\mathrm e^\theta
\\&\phantom{{}={}}{}
-2\mathop{\mathrm{Re}}\bigl(r^{-1}|\chi_{m,n}'|\mathrm e^\theta (H-\lambda)\bigr)
-2\mathop{\mathrm{Re}}\bigl(f^{-1-\rho}r^{-1}\Theta (H-\lambda)\bigr)
.
\end{split}
\label{eq:18062820bb}
\end{align}
Let us combine and bound the first and second terms of \eqref{eq:18062820b} as 
\begin{align}
\begin{split}
&
\bigl(A+\tfrac{\mathrm i}2|\partial f|^2\theta'\bigr)\theta'\Theta \bigl(A-\tfrac{\mathrm i}2|\partial f|^2\theta'\bigr)
+p_jf^{-1}\Theta\ell_{jk}p_k
\\&
\ge 
\tfrac12\bigl(A+\tfrac{\mathrm i}2|\partial f|^2\theta'\bigr)f^{-1}\Theta \bigl(A-\tfrac{\mathrm i}2|\partial f|^2\theta'\bigr)
+\tfrac12p_jf^{-1}\Theta\ell_{jk} p_k
\\&
\ge 
\tfrac12pf^{-1}|\partial f|^2\Theta p
-\tfrac18f^{-1}|\partial f|^4\theta'^2\Theta 
-C_4Q
\\&
\ge 
\tfrac18f^{-1}|\partial f|^4\theta'^2\Theta 
+f^{-1}|\partial f|^2x\Theta 
-C_5Q
+\mathop{\mathrm{Re}}\bigl(f^{-1}|\partial f|^2\Theta (H-\lambda)\bigr).
\end{split}
\label{eq:18062820bbb}
\end{align}
Now by \eqref{eq:18062820b}, \eqref{eq:18062820bb} and \eqref{eq:18062820bbb} 
we obtain 
\begin{align*}
\begin{split}
2\mathop{\mathrm{Im}}\bigl(A\Theta (H-\lambda)\bigr)
&\ge 
\bigl(\tfrac98-C_6f^{-\min\{2,\rho\}}\bigr)\alpha^2f^3r^{-2}\Theta 
+\bigl(\tfrac12-C_6f^{-\min\{2,\rho\}}\bigr) f^{-1}\Theta 
\\&\phantom{{}={}}{}
-C_6\alpha^2 f^{-1}\bigl(\chi_{m-1,m+1}^2+\chi_{n-1,n+1}^2\bigr)\mathrm e^\theta
+\mathop{\mathrm{Re}}\bigl(\gamma_2(H-\lambda)\bigr)
.
\end{split}
\end{align*}
By letting $n_0\in\mathbb N$ large enough 
we obtain the assertion.
\end{proof}

\begin{proof}[Proof of Proposition~\ref{prop:Absence of super-exponentially decaying efb}]
Let $\phi \in L^2_{\mathrm{loc}}(\mathbb R^d)$ and $m_0\in\mathbb N$
satisfy the assumptions of the assertion.
Choose $n_0 \ge 0$ in agreement with Lemma~\ref{lemma:alpha^2-estimate}.
We may let $n_0\ge m_0+3$.
We evaluate the inequality \eqref{eq:keyest} in the state 
$\chi_{m-2,n+2}\phi$,
and then it follows that for any $\alpha\ge 1$ and $n>m\ge n_0$
\begin{align}
 \label{eq:evaluate-in-the-state}
\begin{split}
\bigl\| f^{3/2}r^{-1}\chi_{m,n}^{1/2}\exp(\alpha f^3)\phi \bigr\|^2 
&\le 
C_m\bigl\| \chi_{m-1,m+1}\exp(\alpha f^3)\phi\bigr\|^2 
\\&\phantom{{}={}}{}
+ C_12^{-n/2}\bigl\| \chi_{n-1,n+1}\exp(\alpha f^3)\phi \bigr\|^2.
\end{split}
\end{align}
Since $\bar\chi_{m_0}\exp(\alpha f^3)\phi \in \vB_0^*$ for any $\alpha>0$, 
the second term on the right-hand side of \eqref{eq:evaluate-in-the-state} 
vanishes in the limit $n \to \infty$.
Hence by the Lebesgue monotone convergence theorem we obtain
\begin{equation*}
\bigl\| f^{3/2}r^{-1}\bar\chi_{m}^{1/2}\exp(\alpha f^3)\phi \bigr\|^2 
\le
 C_m\bigl\| \chi_{m-1,m+1}\exp(\alpha f^3)\phi \bigr\|^2,
\end{equation*}
or
\begin{equation} \label{eq:n-infty}
\bigl\| f^{3/2}r^{-1}\bar\chi_{m}^{1/2}\exp\bigl[\alpha (f^3-2^{3(m+2)})\bigr]\phi \bigr\|^2 
\le 
C_m\| \chi_{m-1,m+1}\phi \|^2.
\end{equation}
Now assume $\bar\chi_{m+2}\phi \not\equiv 0$.
The left-hand side of \eqref{eq:n-infty} grows exponentially as $\alpha \to \infty$ whereas the right-hand side remains bounded.
This is a contradiction.
Thus $\bar\chi_{m+2}\phi \equiv 0$. 
Then by Condition~\ref{cond:smooth2wea3n2} we obtain $\phi \equiv 0$ on $\R^d$.
\end{proof}

\section{LAP bounds}\label{sec:181010}

In this section we prove LAP bounds asserted in Theorem~\ref{thm:lap-bounds}. 
Technically, we split $\phi=R(z)\psi$ into two parts according to the size of $f$.
We bound the part of $\phi$ with large $f$ 
employing a commutator computation 
from Lemma~\ref{lemma:commutator-calculation} for a weight 
\begin{equation}\label{eq:1810021842}
\begin{split}
\Theta &=\Theta_{m,\nu}^\delta=\bar\chi_m\theta;\\
\theta&= \theta_\nu^\delta = \int_0^{f/2^\nu}(1+s)^{-1-\delta}\,\mathrm ds 
=\bigl[1-(1+f/2^\nu)^{-\delta}\bigr]\big/\delta
\end{split}
\end{equation}
with parameters $\delta>0$ and $m,\nu\in\N_0$.
On the other hand, the part of $\phi$ with small $f$ can be controlled by 
local compactness
for which we make use of 
Condition~\ref{cond:1803131348b}.
These preliminary arguments are given in Subsection~\ref{subsec:181216},
and the proof of Theorem~\ref{thm:lap-bounds} in Subsection~\ref{subsec:181216b}.

\subsection{Key bounds and local compactness}\label{subsec:181216}

Let us denote the derivatives of functions in $f$ by primes as 
in the previous sections.
Then we have
\begin{equation}\label{eq:lap-theta-derivative}
\theta'=(1+f/2^\nu)^{-1-\delta}\big/2^\nu, \quad
\theta''=-(1+\delta)(1+f/2^\nu)^{-2-\delta}\big/2^{2\nu}.
\end{equation}
The function $\theta$ has the following properties.

\begin{lemma}\label{lem:theta-inequality}
Fix any $\delta>0$ in \eqref{eq:1810021842}.
Then there exist $c, C, C_k>0,\ k=2, 3, \ldots$, such that for any $k=2,3,\ldots$ 
and uniformly in $\nu\in\N_0$
\begin{align*}
&
c/2^\nu \le \theta \le \min\{C, f/2^\nu\},
\\&
c(\min\{2^\nu, f\})^\delta f^{-1-\delta}\theta \le \theta' \le f^{-1}\theta,
\\&
0\le(-1)^{k-1}\theta^{(k)} \le C_kf^{-k}\theta.
\end{align*}
\end{lemma}

We omit the proof of Lemma \ref{lem:theta-inequality}, see e.g. \cite[Lemma~4.2]{IS}.

The following proposition provides 
key bounds for the proof of Theorem~\ref{thm:lap-bounds}.

\begin{proposition}\label{prop:lap-bound-key-estimate}
Assume Condition~\ref{cond:1803131348}.
Let $I\subseteq\R$ be a compact interval, 
fix any $\delta\in(0, \min\{ 2, \rho\})$ in \eqref{eq:1810021842}.
Then there exist $C>0$ and $n \in\mathbb N_0$ such that 
for any $\nu\in\N_0$, $z\in I_\pm$ and $\psi \in \vB$
the states $\phi=R(z)\psi$ satisfy
\begin{align}
\begin{split}
&
\|\theta'^{1/2}\phi\|^2 
+\|(1-x/r)^{1/2}\theta^{1/2}\phi\|_{L^2_{-1/2}}^2 
+ \|\theta'^{1/2}A\phi\|^2 
+ \langle p_jf^{-1}\theta \ell_{jk}p_k\rangle_\phi 
\\&\le 
C\left( \|\phi\|_{\vB^*} \|\psi\|_\vB 
+ \|A\phi\|_{\vB^*} \|\psi\|_\vB 
+ \|\chi_n\theta^{1/2}\phi\|^2 \right).
\end{split}
\label{eq:181218}
\end{align}
\end{proposition}
\begin{proof}
Fix $I$ and $\delta$ as in the assertion.
We choose $m\in\mathbb N_0$ in \eqref{eq:1810021842} large enough that \eqref{eq:190315} holds.
It suffices to show that 
there exist $c_1, C_1>0$ and $n\in\mathbb N_0$ such that uniformly in $z\in I_\pm$ 
and $\nu\in\N_0$
\begin{equation}
\begin{split}
\mathop{\mathrm{Im}}\bigl( A\Theta(H-z) \bigr) 
&\ge 
c_1\theta' +c_1f^{-1}(1-x/r)\theta +c_1A\theta'A +c_1p_jf^{-1}\theta\ell_{jk} p_k 
\\&\phantom{{}={}}{}
- C_1\chi_n^2\theta 
+\mathop{\mathrm{Re}}\bigl( \gamma_1(H-z) \bigr),
\end{split}
\label{eq:181218b}
\end{equation} 
where $\gamma_1=\gamma_{1,z,\nu}$ is a certain uniformly bounded complex-valued function: 
$|\gamma_1| \le C_1$. 
In fact, deduction of \eqref{eq:181218} from \eqref{eq:181218b} is 
straightforward by taking expectation of \eqref{eq:181218b} in the state $\phi=R(z)\psi$.
Hence we prove \eqref{eq:181218b} in what follows.

By Lemmas~\ref{lemma:commutator-calculation}, \ref{lem:theta-inequality} 
and the Cauchy--Schwarz inequality 
we can bound uniformly in $z=\lambda\pm \mathrm i\Gamma\in I_\pm$ and $\nu\in\N_0$
\begin{align}
\begin{split}
&
2\mathop{\mathrm{Im}}(A\Theta(H-z))
\\&\ge 
A\Theta' A
+p_jf^{-1}\Theta \ell_{jk}p_k
+p_jf^{-1}|\partial f|^2\bigl(\delta_{jk}-(\partial_jr)(\partial_kr)\bigr)\Theta p_k
\\&\phantom{{}={}}{}
+\tfrac12 f^{-1}(1-x/r)\Theta 
-\tfrac14|\partial f|^4\Theta'''
-\tfrac12(\partial^f|\partial f|^2)\Theta''
-\tfrac12f^{-1}|\partial f|^4\Theta'' 
\\&\phantom{{}={}}{}
+ q_4\Theta'
+ q_5\Theta 
-2\mathop{\mathrm{Im}}\bigl(q_2\Theta p^f\bigr)
-2\mathop{\mathrm{Re}}\bigl(f^{-1}|\partial f|^2\Theta H\bigr)
-\mathop{\mathrm{Re}}\bigl(|\partial f|^2\Theta' H\bigr)
\\&\phantom{{}={}}{}
-2\lambda \mathop{\mathrm{Im}}(A\Theta)
\mp 2\Gamma \mathop{\mathrm{Re}}(A\Theta)
\\&\ge 
\tfrac12A\Theta' A
+\tfrac12p_jf^{-1}\Theta \ell_{jk}p_k
+\tfrac12p|\partial f|^2\Theta' p
+\tfrac12 f^{-1}(1-x/r)\Theta 
-C_2Q
\\&\phantom{{}={}}{}
-2\mathop{\mathrm{Re}}\bigl(f^{-1}|\partial f|^2\Theta (H-z)\bigr)
-\mathop{\mathrm{Re}}\bigl(|\partial f|^2\Theta' (H-z)\bigr)
\mp 2\Gamma \mathop{\mathrm{Re}}(\Theta p^f)
,
\end{split}
\label{eq:181217}
\end{align}
where $Q$ is an \emph{admissible} error of the form
\begin{equation*}
Q = f^{-1-\min\{2,\rho\}}\theta + pr^{-1}f^{-1-\rho}\theta p.
\end{equation*}
We rewrite and bound the third term on the right-hand side of \eqref{eq:181217} as 
\begin{equation}\label{eq:1810022054}
\begin{split}
\tfrac12p|\partial f|^2\Theta'p 
&\ge  
\tfrac14p|\partial f|^2\Theta'p 
\\&= 
\tfrac12\mathop{\mathrm{Re}}\bigl(|\partial f|^2\Theta'(H-z)\bigr) 
+\tfrac12(x-q+\lambda)|\partial f|^2\Theta' 
+\tfrac18(\Delta|\partial f|^2\Theta')
\\&\ge 
\tfrac12\mathop{\mathrm{Re}}\bigl(|\partial f|^2\Theta'(H-z)\bigr) 
+\tfrac14 r^{-1}x\Theta' 
-C_3Q.
\end{split}
\end{equation}
As for the eighth term of \eqref{eq:181217}, 
we use the Cauchy-Schwarz inequality 
and Lemma~\ref{lem:theta-inequality}, 
and then
\begin{equation}
\begin{split}\label{eq:1810022238}
\mp 2\Gamma \mathop{\mathrm{Re}}(\Theta p^f)
&
\ge 
-C_4\Gamma p r^{-1}\Theta p 
-C_5\Gamma 
\\&
\ge 
-2C_4\Gamma\mathop{\mathrm{Re}}\bigl( r^{-1}\Theta (H-z) \bigr) 
- C_6\Gamma
\\&\ge 
-2C_4\Gamma\mathop{\mathrm{Re}}\bigl( r^{-1}\Theta (H-z) \bigr) 
\pm C_6\mathop{\mathrm{Im}}(H-z)
.
\end{split}
\end{equation}
By \eqref{eq:181217}, \eqref{eq:1810022054} and \eqref{eq:1810022238}
we obtain
\begin{equation}\label{eq:1810022258}
\begin{split}
2\mathop{\mathrm{Im}}(A\Theta(H-z))
&\ge 
\tfrac12A\theta' A
+\tfrac12p_jf^{-1}\theta \ell_{jk}p_k
+\tfrac14f^{-1}(1-x/r)\theta 
+\tfrac14 \theta' 
\\&\phantom{{}={}}{}
-C_7Q
+\mathop{\mathrm{Re}}\bigl(\gamma_2(H-z)\bigr) 
.
\end{split}
\end{equation}
Finally we combine and bound the fourth and fifth terms of \eqref{eq:1810022258} 
as, for large $n\in\mathbb N_0$,
\begin{equation}\label{eq:1810022302}
\tfrac14\theta' - C_7Q 
\ge 
\tfrac18\theta' 
- C_8\chi_n^2\theta 
- C_9\mathop{\mathrm{Re}}\bigl( r^{-1}f^{-1-\rho}\theta(H-z) \bigr).
\end{equation}
Hence by \eqref{eq:1810022258} and \eqref{eq:1810022302} the assertion follows.
\end{proof}

For the proof of Theorem~\ref{thm:lap-bounds}
we also use local compactness of the following form.

\begin{proposition}\label{prop:lap-bound-key-estimateb}
Assume Condition~\ref{cond:1803131348b}. 
Then for any $l\in\mathbb N_0$ and compact interval 
$I$ the mapping 
\begin{align*}
  \chi_lP_H(I)\colon L^2\to L^2
\end{align*}
is compact, where $P_H(I)$ denotes the spectral projection for $H$ onto $I$. 
\end{proposition}
\begin{proof}
Fix any $l\in\mathbb N_0$
and any compact interval $I$.
We let $\{\psi_k\}_{k\in\mathbb N_0}\subseteq L^2$ be a bounded sequence,
and set $\phi_k=\chi_lP_H(I)\psi_k$. 
First, using Condition~\ref{cond:1803131348b}
we have 
\begin{align*}
\|\phi_k\|^2
+\|p\phi_k\|^2
&\le 
\|\phi_k\|^2
+2\langle H\rangle_{\phi_k}
-2\langle -x+q\rangle_{\phi_k}
\\&\le 
C_1
\|\phi_k\|^2
+
C_2\langle H\rangle_{\phi_k}
\\&
\le C_3\|\psi_k\|^2.
\end{align*}
Hence the sequence $\{\phi_k\}_{k\in\mathbb N_0}$ is bounded 
in the standard Sobolev space $H^1(\mathbb R^d)$.
Then by Rellich's compact embedding theorem
and the diagonal argument
it suffices to show that 
\begin{align}
\lim_{\nu\to\infty}\sup_k\|\eta_\nu\phi_k\|=0;\quad
\eta_\nu(x,y)=1-\chi(-x/2^\nu),
\label{eq:18121810}
\end{align}
see \eqref{eq:chidef} for $\chi$.
Let $\epsilon>0$. 
Using again  Condition~\ref{cond:1803131348b} we deduce that 
 for any large $\nu\in\mathbb N_0$,  independent of $k\in\mathbb N_0$, 
\begin{align*}
\|\eta_\nu\phi_k\|^2
&\le 
\epsilon\langle -x+q\rangle_{\eta_\nu\phi_k} 
\le 
\epsilon\langle H\rangle_{\eta_\nu\phi_k} 
\le 
C_4\epsilon\|\psi_k\|^2
,
\end{align*}
where $C_4> 0$ does not depend on $\epsilon>0$ or $k\in\mathbb N_0$.
This verifies \eqref{eq:18121810}, and hence we are done.
\end{proof}

\subsection{Proof of LAP bounds}\label{subsec:181216b}

Now we prove Theorem~\ref{thm:lap-bounds} 
employing Propositions~\ref{prop:lap-bound-key-estimate},  
\ref{prop:lap-bound-key-estimateb} and a contradiction argument.

\begin{proof}
Let $I $ be a compact interval.

\smallskip
\noindent
{\it Step 1.}
First we reduce the proof of Theorem~\ref{thm:lap-bounds} to the single bound
\begin{equation}\label{eq:lap-single-bound}
\|\phi\|_{\vB^*} \le C_1\|\psi\|_\vB;\quad 
\phi=R(z)\psi.
\end{equation}
Assume \eqref{eq:lap-single-bound} holds true.
Fix any $\delta\in(0,\min\{2,\rho\})$.
Then by Proposition~\ref{prop:lap-bound-key-estimate} and \eqref{eq:lap-single-bound} 
there exists $C_2>0$ such that uniformly in $\epsilon_1\in(0,1)$ and $\nu\in\N_0$
\begin{equation}\label{eq:1810031651}
\begin{split}
&
\|(1-x/r)^{1/2}\theta^{1/2}\phi\|_{L^2_{-1/2}}^2 
+\|\theta'^{1/2}A\phi\|^2 
+ \langle p_jf^{-1}\theta \ell_{jk}p_k \rangle_{\phi} 
\\&\le 
\epsilon_1\|A\phi\|_{\vB^*}^2 
+ \epsilon_1^{-1}C_2\|\psi\|_\vB^2.
\end{split}
\end{equation}
For each $\nu\ge 0$, 
restricting the integral region to $\{2^\nu\le f<2^{\nu+1}\}$, 
we can bound the second term on the left-hand side of \eqref{eq:1810031651} as
\begin{equation}\label{eq:1810031653}
\|\theta'^{1/2}A\phi \|^2 
\ge 
3^{-(1+\delta)}2^{-\nu}\|F_\nu A\phi\|^2, 
\end{equation}
where $F_\nu$ is from \eqref{eq:190316}.
As for the first and third terms on the same side, 
letting $\nu=0$ and using Lemma~\ref{lem:theta-inequality}, we have
\begin{equation}\label{eq:1810031700}
\begin{split}
&\|(1-x/r)^{1/2}\theta^{1/2}\phi\|_{L^2_{-1/2}}^2 
+\langle p_jf^{-1}\theta \ell_{jk}p_k \rangle_{\phi} 
\\&\ge 
c_1\|(1-x/r)^{1/2}\phi\|_{L^2_{-1/2}}^2 
+c_1\langle p_jf^{-1}\ell_{jk}p_k \rangle_{\phi}.
\end{split}
\end{equation}
We use \eqref{eq:1810031653} and \eqref{eq:1810031700} separately in  
 \eqref{eq:1810031651}. The bound with the right-hand side of  \eqref{eq:1810031700} is
 independent of $\nu$, and for the bound with  the right-hand side of
 \eqref{eq:1810031653} we take  the supremum in $\nu\in\N_0$.
Then we obtain uniformly in $\epsilon_1\in(0,1)$ 
\begin{equation*}
c_2\|(1-x/r)^{1/2}\phi\|_{L^2_{-1/2}}^2 
+c_2\|A\phi\|_{\vB^*}^2 
+ c_2\langle p_jf^{-1}\ell_{jk}p_k \rangle_{\phi} 
\le 
\epsilon_1\|A\phi\|_{\vB^*}^2 
+\epsilon_1^{-1}C_2\|\psi\|_\vB^2.
\end{equation*}
Therefore by letting $\epsilon_1\in(0, c_2/2)$ it follows that
\begin{equation*}
\|(1-x/r)^{1/2}\phi\|_{L^2_{-1/2}}^2 
+\|p^f\phi\|_{\vB^*} 
+ \langle p_jf^{-1}\ell_{jk}p_k \rangle^{1/2}_\phi 
\le 
C_3\|\psi\|_\vB.
\end{equation*}
Hence Theorem~\ref{thm:lap-bounds} reduces to 
the single bound \eqref{eq:lap-single-bound}.

\smallskip
\noindent
{\it Step 2.}
Next we prove \eqref{eq:lap-single-bound} arguing by 
contradiction. 
So assume there exist $z_k\in I_\pm$ and $\psi_k\in\vB$ such that 
\begin{equation}\label{eq:psi_k}
\lim_{k \to \infty}\|\psi_k\|_\vB=0, 
\quad 
\|\phi_k\|_{\vB^*}=1;\ \ 
\phi_k=R(z_k)\psi_k.
\end{equation}
 By the time-reversal property we may assume that  $z_k\in I_+$. 
In addition,
by choosing a subsequence 
we may assume that $z_k$ converges to some $z\in \overline I_+$.
If $\mathop{\mathrm{Im}}z>0$, then 
\eqref{eq:psi_k} contradicts the bounds
\begin{equation*}
\|\phi_k\|_{\vB^*} 
\le 
\|R(z_k)\psi_k\| 
\le 
\|R(z_k)\|_{\vL(L^2)}\|\psi_k\|
\le 
\|R(z_k)\|_{\vL(L^2)}\|\psi_k\|_\vB
\end{equation*}
as $k\to\infty$.
Hence  we have a  real limit
\begin{equation}\label{eq:lim-z_k}
\lim_{k\to\infty}z_k = z = \lambda\in I.
\end{equation}
Let $s>1/2$.
By choosing a further subsequence we may assume 
$\phi_k$ converges weakly to some $\phi\in L^2_{-s}$. 
Then, in fact, $\phi_k$ converges strongly to $\phi\in L^2_{-s}$.
To verify this let us fix $s'\in(1/2,s)$ and $h\in C_{\rm c}^\infty(\R)$ 
with $h=1$ on a neighborhood of $I$, and decompose for any $l\in\mathbb N_0$
\begin{align}
\begin{split}
f^{-s}\phi_k 
&= 
(f^{-s}h(H))(\chi_l f^s)(f^{-s}\phi_k) 
+ (f^{-s}h(H)f^s)(\bar\chi_l f^{s'-s})(f^{-s'}\phi_k) 
\\&\phantom{{}={}}{}
+ f^{-s}(1-h(H))R(z_k)\psi_k.
\end{split}
\label{eq:18121814}
\end{align}
By \eqref{eq:psi_k} we see that the last term on the right-hand side 
of \eqref{eq:18121814} converges to $0$ in $L^2$.
Since $f^{-s}h(H)f^s$ is a bounded operator, 
by choosing $l\in\mathbb N_0$ 
sufficiently large the second term of \eqref{eq:18121814}
can be arbitrarily small in $L^2$.
Lastly, since $f^{-s}h(H)$ is compact by 
Proposition~\ref{prop:lap-bound-key-estimateb}, 
the first term of \eqref{eq:18121814} converges strongly in $L^2$.
Therefore $\phi_k$ converges to $\phi$ in $L^2_{-s}$:
\begin{equation}\label{eq:lim-phi_k}
\lim_{k\to\infty}\phi_k 
= \phi \ \text{ in } L^2_{-s}.
\end{equation}
By \eqref{eq:psi_k}, \eqref{eq:lim-z_k} and \eqref{eq:lim-phi_k} it follows that
\begin{equation}\label{eq:1810031853}
(H-\lambda)\phi = 0 \ \text{ in the distributional sense}.
\end{equation}
In addition, we can verify $\phi\in\vB_0^*$.
In fact, let us fix $\delta=2s-1>0$ and apply Proposition~\ref{prop:lap-bound-key-estimate} 
to $\phi_k$.
Then, letting $k\to\infty$ and using \eqref{eq:psi_k}, \eqref{eq:lim-phi_k} and 
Lemma~\ref{lem:theta-inequality}, we obtain for all $\nu\in\N_0$
\begin{equation}\label{eq:1810031903}
\|\theta'^{1/2}\phi\| 
\le 
\|\chi_n\theta^{1/2}\phi\| 
\le 
C_42^{-\nu/2}\|\chi_nf^{1/2}\phi\|.
\end{equation}
Then by letting $\nu\to\infty$ in \eqref{eq:1810031903} 
we obtain $\phi\in\vB_0^*$.
Therefore by \eqref{eq:1810031853} and Theorem~\ref{thm:priori-decay-b_0}, 
we have $\phi=0$, but this is a contradiction.
In fact, we can prove, as in Step 1, 
\begin{align*}
1=\|\phi_k\|^2_{\mathcal B^*}
\le C_5\bigl(\|\psi_k\|_{\mathcal B}^2+\|\chi_n\phi_k\|^2\bigr).
\end{align*}
But the right-hand side can be made arbitrarily small (in particular
smaller than $1$) by taking $k$ big enough.
\end{proof}

\section{Radiation condition bounds}\label{sec:181207}

Here we prove Theorem~\ref{thm:RC-bound}
and Corollaries~\ref{cor:Limiting-Absorption-Principle-Stark}, \ref{cor:RC-bound-real}, 
and \ref{cor:Sommerfeld-unique-result}. 
For simplicity of arguments we prove the assertions only for the upper sign.
For the proof of Theorem~\ref{thm:RC-bound} 
the  form inequality \eqref{eq:181224} below is a key ingredient, 
cf.\ \eqref{eq:14.9.26.9.53ffaabb}, \eqref{eq:keyest} and \eqref{eq:181218b} 
in the former sections.

In this section we always assume Condition~\ref{cond:additional-condition}. 
Furthermore, we throughout the section fix a compact interval $I$ 
and $l\in\mathbb N_0$ as in \eqref{eq:phase-a},
so that the phase $a$ is always a fixed function. 
We may let $l\in\mathbb N_0$ be large without loss of generality,
so that the formulas from \eqref{180814} and \eqref{eq:180816} are available 
on $\mathop{\mathrm{supp}}a$,
and also that $\mathop{\mathrm{supp}}a\cap \mathop{\mathrm{supp}}q_3=\emptyset$.

\subsection{Key bounds}

We first present basic properties of $a$.

\begin{lemma}\label{lem:1812101541}
\begin{enumerate}
\item
There exists $C>0$ such that for any $z\in I_+$ and $(x,y)\in \mathbb R^d$ 
\begin{align*}
&|a| \le C, 
\quad 
\mathop{\mathrm{Im}}a\ge \tfrac14\bar\chi_lfr^{-2}, 
\quad 
|\widetilde\partial a|\le Cf^{-1-\min\{2,\rho,\tilde\rho\}}r^{-1}
,
\end{align*}
where $\widetilde\partial$ is from \eqref{eq:18121512}.

\item
Let $m\in\mathbb N_0$ with $m\ge l+2$.
Then for any $z\in I_+$ one can write 
\begin{align*}
\bar\chi_m(H-z )
&= 
\bar\chi_m
\bigl[
(A+ a)r(A- a) 
+p_jr\ell_{jk}p_k 
+r
-x 
+ q_6\bigr]
\end{align*}
with 
\begin{align}
q_6 
&=
\bigl(p^far\bigr)
+ra^2
-r+q-z
+\tfrac14r(\Delta f)^2
+\tfrac12\bigl(\partial^fr\Delta f\bigr)
.
\label{eq:190322}
\end{align}
The function $q_6$ in particular satisfies for some $C'>0$
\begin{equation}
\bar\chi_m|q_6|\le C'\bar\chi_mf^{-1-\min\{2,\rho,\tilde\rho\}}.
\label{eq:190322b}
\end{equation}
\end{enumerate}
\end{lemma}
\begin{proof}
The bounds in (1) follow from straightforward computations,  
and here we only do  (2).
Using the formulas from \eqref{180814} and \eqref{eq:180816},
we can rewrite 
\begin{align*}
\begin{split}
\bar\chi_m(H-z )
&= 
\bar\chi_m\bigl[(p^f)^*rp^f 
+p_jr\ell_{jk}p_k 
-x+q-z\bigr] 
\\&
= 
\bar\chi_m\bigl[ArA
+p_jr\ell_{jk}p_k 
-x+q-z 
+\tfrac14r(\Delta f)^2
+\tfrac12(\partial^fr\Delta f)
\bigr]
\\&
= 
\bar\chi_m\bigl[(A+ a)r(A- a)
+p_jr\ell_{jk}p_k 
+r-x+q_6
\bigr]
\end{split}
\end{align*}
with $q_6$ given as \eqref{eq:190322}.
The last two terms of \eqref{eq:190322} obviously satisfy \eqref{eq:190322b}. 
In addition  we can compute on $\mathop{\mathrm{supp}}\bar\chi_m$, 
using the formulas from \eqref{180814} and \eqref{eq:180816}, 
\begin{align*}
\bigl(p^far\bigr)+a^2r-r+q-z
&=
\tfrac{\mathrm i}4\bigl(zfr^{-2}+2(\partial^fq_1)-fr^{-2}q_1\bigr)\big/\sqrt{(z+r-q_1)/r}
\\&\phantom{{}={}}{}
+\tfrac18r^{-2}
-\tfrac3{16}f^2r^{-3}
+q_2
.
\end{align*}
Hence we can verify \eqref{eq:190322b}.
\end{proof}

We will employ the following weight functions:
\begin{align}\label{eq:theta-rc}
\begin{split}
\Theta&=\Theta_{m,\nu}^{\beta,\delta}=\bar\chi_m\theta^{2\beta};\\
\theta &= \theta_\nu^\delta 
= \int_0^{f/2^\nu}(1+s)^{-1-\delta}\,\mathrm ds 
= \bigl[1-(1+f/2^\nu)^{-\delta}\bigr]\big/\delta
\end{split}
\end{align}
with parameters $\beta,\delta>0$ and $m,\nu\in\mathbb N_0$.
Note that $\theta$ is the same as that in Section~\ref{sec:181010},
and hence Lemma~\ref{lem:theta-inequality} is available.
We denote derivatives in $f$ by primes as in \eqref{eq:lap-theta-derivative}.

\begin{lemma}\label{lem:rc-bound-key-estimate}
Let $\beta\in(0,1/2)$. 
Fix any $m\in\mathbb N_0$ with $m\ge l+2$, 
and fix any $\delta>0$ in \eqref{eq:theta-rc}.
Then there exist $c,C>0$ such that uniformly in $z\in I_+$ and $\nu\in\N_0$, 
as quadratic forms on $\vD(H)$,
\begin{equation}
\begin{split}
&\mathop{\mathrm{Im}}\bigl((A-a)^*\Theta(H-z)\bigr) 
\\&\ge 
cf^{-1}(1-x/r)\Theta
+c(A-a)^*\bar\chi_m\theta'\theta^{2\beta-1}(A-a) 
+ cp_jf^{-1}\Theta\ell_{jk}p_k 
\\&\phantom{{}={}}{}
- Cf^{-1-\min\{4,2\rho,2\tilde\rho\}+\delta}\theta^{2\beta} 
+\mathop{\mathrm{Re}}\bigl(\gamma\theta^{2\beta}(H-z)\bigr),
\end{split}
\label{eq:181224}
\end{equation}
where $\gamma=\gamma_{z,\nu}$ is a certain function 
satisfying $|\gamma|\le Cf^{-1-\min\{4,2\rho,2\tilde\rho\}+\delta}$.
\end{lemma}
\begin{proof}
In this proof we repeatedly use the formulas from 
\eqref{180814} and \eqref{eq:180816} without mentioning.
Fix $\beta\in(0,1/2)$, $m\in\mathbb N_0$ and $\delta>0$ 
as in the assertion.
By Lemmas~\ref{lem:1812101541} we write 
\begin{align}
\begin{split}
&
2\mathop{\mathrm{Im}}\bigl((A-a)^*\Theta(H-z)\bigr) 
\\&= 
2\mathop{\mathrm{Im}}\bigl((A-a)^*\Theta(A+a)r(A-a)\bigr) 
+ 2\mathop{\mathrm{Im}}\bigl((A-a)^*\Theta p_jr\ell_{jk}p_k\bigr) 
\\&\phantom{{}={}}{}
+ 2\mathop{\mathrm{Im}}\bigl((A-a)^*\Theta(r-x)\bigr) 
+ 2\mathop{\mathrm{Im}}\bigl((A-a)^*\Theta q_6\bigr) ,
\end{split}
\label{eq:rc-bound-key-estimateb}
\end{align}
and we further compute each term on the right-hand side of \eqref{eq:rc-bound-key-estimateb}. 
All the estimates below are uniformly in $z\in I_+$ and $\nu\in\N_0$.

By Lemma~\ref{lem:1812101541}
the first term of \eqref{eq:rc-bound-key-estimateb} can be computed and 
bounded as 
\begin{align}
\begin{split}
&
2\mathop{\mathrm{Im}}\bigl((A-a)^*\Theta (A+a)r(A-a)\bigr) 
\\&
=
(A-a)^*(\partial^f\bar\chi_m)\theta^{2\beta}r(A-a)
+2\beta (A-a)^*\bar\chi_mr|\partial f|^2\theta'\theta^{2\beta-1}(A-a)
\\&\phantom{{}={}}{}
-(A-a)^*(\partial^fr)\Theta (A-a)
+2(A-a)^*(\mathop{\mathrm{Im}}a)r\Theta (A-a)
\\&
\ge 
\beta (A-a)^*\bar\chi_m\theta'\theta^{2\beta-1}(A-a)
.
\label{eq:rc-bound-key-estimatebb}
\end{split}
\end{align}
For the second term of \eqref{eq:rc-bound-key-estimateb} 
we use Lemma~\ref{lem:1812101541},
the Cauchy--Schwarz inequality and Lemma~\ref{lem:theta-inequality}.
Omitting some computations, we finally obtain 
\begin{align}
\begin{split}
&
2\mathop{\mathrm{Im}}\bigl((A-a)^*\Theta p_jr\ell_{jk}p_k\bigr) 
\\&= 
2\mathop{\mathrm{Im}}\bigl(p_jAr\Theta\ell_{jk}p_k\bigr) 
+2\mathop{\mathrm{Im}}\bigl([A,p_j]r\Theta\ell_{jk}p_k\bigr) 
-2\mathop{\mathrm{Im}}\bigl(p_ja^*r\Theta\ell_{jk}p_k\bigr) 
\\&\phantom{{}={}}{}
-2\mathop{\mathrm{Im}}\bigl([a^*,p_j]r\Theta\ell_{jk}p_k\bigr) 
\\&\ge 
p_jf^{-1}\Theta\ell_{jk}p_k
-\beta p_j\bar\chi_m\theta'\theta^{2\beta-1}\ell_{jk}p_k 
+\tfrac14pfr^{-2}\Theta p 
-\tfrac12(p^r)^*f^{-1}r^{-1}\Theta p^r 
\\&\phantom{{}={}}{}
-\mathop{\mathrm{Im}}\bigl((\partial_j\Delta f)r\Theta\ell_{jk}p_k\bigr) 
-2\mathop{\mathrm{Re}}\bigl((\partial_ja^*)r\Theta\ell_{jk}p_k\bigr) 
-C_1Q
\\&
\ge 
(1-\beta-\epsilon_1)p_jf^{-1}\Theta\ell_{jk}p_k
-\tfrac14pf^{-1}r^{-2}(r-x)\Theta p
-C_2\epsilon_1^{-1}Q
, 
\end{split}
\label{eq:rc-bound-key-estimatebbb}
\end{align}
where $\epsilon_1\in (0,1)$ is a small constant fixed below,
$C_2>0$ is independent of $\epsilon_1$,
and $Q$ is an \emph{admissible} error of the form
\begin{equation*}
Q = 
f^{-1-\min\{4,2\rho,2\tilde\rho\}+\delta}\theta^{2\beta} 
+ pf^{-1-\min\{4,2\rho,2\tilde\rho\}+\delta}r^{-1}\theta^{2\beta}p
.
\end{equation*}
As for the third term of \eqref{eq:rc-bound-key-estimateb},
we simply compute and bound it by Lemma~\ref{lem:theta-inequality} as 
\begin{align}
\begin{split}
&
2\mathop{\mathrm{Im}}\bigl((A-a)^*\Theta(r-x)\bigr) 
\\&
\ge 
-\beta\bar\chi_mr^{-1}(r-x) \theta'\theta^{2\beta-1}
+\tfrac12fr^{-2}(r-x)\Theta 
-C_3Q
\\&
\geq 
(\tfrac12-\beta)f^{-1}(1-x/r)\Theta 
+\tfrac12xf^{-1}r^{-2}(r-x)\Theta 
-C_3Q.
\end{split}
\label{eq:rc-bound-key-estimatebbbb}
\end{align}
The last term of \eqref{eq:rc-bound-key-estimateb} is bounded by 
using the Cauchy--Schwarz inequality and Lemmas~\ref{lem:1812101541} as 
\begin{align}
\begin{split}
&
2\mathop{\mathrm{Im}}\bigl((A-a)^*\bar\chi_m\theta^{2\beta}q_6\bigr) 
\ge 
-\epsilon_1(A-a)^*\bar\chi_mf^{-1-\delta}\theta^{2\beta}(A-a)
-C_4\epsilon_1^{-1}Q.
\end{split}
\label{eq:rc-bound-key-estimate}
\end{align}

By 
\eqref{eq:rc-bound-key-estimateb}, \eqref{eq:rc-bound-key-estimatebb}, 
\eqref{eq:rc-bound-key-estimatebbb}, \eqref{eq:rc-bound-key-estimatebbbb}
and \eqref{eq:rc-bound-key-estimate}
we have 
\begin{align}
\begin{split}
&
2\mathop{\mathrm{Im}}\bigl((A-a)^*\Theta (H-z)\bigr) 
\\&
\ge 
(\tfrac12-\beta)f^{-1}(1-x/r)\Theta 
+(A-a)^* 
\bigl(\beta\theta'-\epsilon_1f^{-1-\delta}\theta\bigr)\bar\chi_m\theta^{2\beta-1} (A-a)
\\&\phantom{{}={}}{}
+(1-\beta-\epsilon_1)p_j f^{-1}\Theta\ell_{jk}p_k
-\tfrac14p f^{-1}r^{-2}(r-x)\Theta p
\\&\phantom{{}={}}{}
+\tfrac12 xf^{-1}r^{-2}(r-x)\Theta 
-C_5\epsilon_1^{-1}Q
.
\end{split}
\label{eq:1812112355}
\end{align}
The first term on the right-hand side of 
\eqref{eq:1812112355} conform with the assertion,
and so do the second and third 
by using Lemma~\ref{lem:theta-inequality} and choosing small $\epsilon_1\in(0,1)$.
Let us combine the fourth and fifth terms of \eqref{eq:1812112355} as 
\begin{align}
\begin{split}
&-\tfrac14p f^{-1}r^{-2}(r-x)\Theta p
+\tfrac12 xf^{-1}r^{-2}(r-x)\Theta 
\\&
\ge 
-C_6 f^{-1}r^{-1}(1-x/r)\Theta 
-\tfrac12\mathop{\mathrm{Re}}\bigl( f^{-1}r^{-2}(r-x)\Theta (H-z)\bigr)
-C_6Q.
\end{split}\label{eq:1812120010b}
\end{align}
Finally we bound the remainder term $Q$ as
\begin{equation}\label{eq:1812120010}
\begin{split}
- Q 
&\ge 
- C_7f^{-1-\min\{4,2\rho,2\tilde\rho\}+\delta}\theta^{2\beta} 
- 2\mathop{\mathrm{Re}}\bigl(f^{-1-\min\{4,2\rho,2\tilde\rho\}+\delta}r^{-1}\theta^{2\beta}(H-z)\bigr).
\end{split}
\end{equation}
Hence by \eqref{eq:1812112355}, \eqref{eq:1812120010b} 
and \eqref{eq:1812120010} the assertion follows.
\end{proof}

\subsection{Proof of radiation condition bounds}

Here we prove the radiation condition bounds, Theorem~\ref{thm:RC-bound}.

\begin{proof}[Proof of Theorem~\ref{thm:RC-bound}]
For $\beta=0$ the assertion is obvious by Theorem~\ref{thm:lap-bounds}.
Hence we may let $\beta\in(0,\beta_c)$.
Take any $m\ge l+2$ and $\delta\in(0, \min\{4,2\rho,2\tilde\rho\}-2\beta)$,
and apply Lemma~\ref{lem:rc-bound-key-estimate}
to 
 the state $\phi=R(z)\psi$ with $\psi\in f^{-\beta}\vB$ and $z\in I_+$.
Then by the Cauchy--Schwarz inequality, 
Theorem~\ref{thm:lap-bounds} and Lemma~\ref{lem:theta-inequality}
\begin{equation}\label{eq:1811291354}
\begin{split}
&
\bigl\|(1-x/r)^{1/2}\Theta^{1/2}\phi\bigr\|_{L^2_{-1/2}}^2 
+\bigl\|\bar\chi_m^{1/2}\theta'^{1/2}\theta^{\beta-1/2}(A-a)\phi\bigr\|^2 
+\bigl\langle p_jf^{-1}\Theta\ell_{jk}p_k \bigr\rangle_{\phi} 
\\&\le 
C_1\Bigl[ 
\bigl\|\Theta^{1/2}(A-a)\phi\bigr\|_{\vB^*}
\|\theta^\beta\psi\|_\vB 
+\bigl\|f^{-(1+\min\{4,2\rho,2\tilde\rho\}-\delta)/2}\theta^\beta \phi\bigr\|^2
\\&\phantom{{}={}C_1\Bigl[ }{}
+
\bigl\|f^{-(1+\min\{4,2\rho,2\tilde\rho\}-\delta)/2}\theta^\beta \phi\bigr\| 
\bigl\|f^{-(1+\min\{4,2\rho,2\tilde\rho\}-\delta)/2}\theta^\beta\psi\bigr\|
\Bigr]
\\&\le 
C_22^{-2\beta\nu}\Bigl[ 
\bigl\|\bar\chi_m^{1/2}f^\beta(A-a)\phi\bigr\|_{\vB^*}\|f^\beta\psi\|_\vB 
+ \|f^\beta\psi\|_\vB^2\Bigr].
\end{split}
\end{equation}
Here we have $f^\beta(A-a)\phi=f^\beta(A-a)R(z)\psi\in\vB^*$ for each
$z\in I_+$  (seen 
by commuting $R(z)$ and powers of $f$)
and hence the right-hand side of \eqref{eq:1811291354} is finite.
Then it follows by \eqref{eq:1811291354} that 
\begin{equation}\label{eq:1811291412}
\begin{split}
&
2^{2\beta\nu}\bigl\|(1-x/r)^{1/2}\Theta^{1/2}\phi\bigr\|_{L^2_{-1/2}}^2 
+ 2^{2\beta\nu}\bigl\|\bar\chi_m^{1/2}\theta'^{1/2}\theta^{\beta-1/2}(A-a)\phi\bigr\|^2 
\\&
+ 2^{2\beta\nu}\bigl\langle p_jf^{-1}\Theta\ell_{jk}p_k \bigr\rangle_{\phi} 
\\&\le 
C_2\Bigl[ 
\bigl\|\bar\chi_m^{1/2}f^\beta(A-a)\phi\bigr\|_{\vB^*}
\|f^\beta\psi\|_\vB 
+ \|f^\beta\psi\|_\vB^2\Bigr].
\end{split}
\end{equation}
In the second term on the left-hand side of \eqref{eq:1811291412}
restrict the integral region to $\{2^\nu\le f<2^{\nu+1}\}$ 
and take supremum in $\nu\in\N_0$, and then we obtain
\begin{equation*}
c_1\bigl\|\bar\chi_m^{1/2}f^\beta(A-a)\phi\bigr\|_{\vB^*}^2 
\le 
C_2\Bigl[ \bigl\|\bar\chi_m^{1/2}f^\beta(A-a)\phi\bigr\|_{\vB^*}\|f^\beta\psi\|_\vB 
+ \|f^\beta\psi\|_\vB^2\Bigr].
\end{equation*}
By the Cauchy--Schwarz inequality this implies
\begin{equation}\label{eq:rc-bound-first-term}
\bigl\|\bar\chi_m^{1/2}f^\beta(A-a)\phi\bigr\|_{\vB^*} 
\le 
C_3\|f^\beta\psi\|_\vB.
\end{equation}
As for the first and third terms on the left-hand side of \eqref{eq:1811291412} 
we  first bound
$\theta^{2\beta}\geq ( f\theta')^{2\beta}$,  then take the limit
$\nu\to\infty$ using 
the Lebesgue monotone convergence theorem, and use  
\eqref{eq:rc-bound-first-term} to estimate the right-hand side, yielding
\begin{equation}\label{eq:rc-bound-first-termb}
\bigl\|\bar\chi_m^{1/2}f^\beta(1-x/r)^{1/2}\phi\bigr\|_{L^2_{-1/2}}^2 
+\bigl\langle p_j\bar\chi_m^{1/2}f^{2\beta-1}\ell_{jk}p_k \bigr\rangle_{\phi} 
\le 
C_4\|f^\beta\psi\|_\vB^2.
\end{equation}
 From \eqref{eq:rc-bound-first-term} and \eqref{eq:rc-bound-first-termb} we 
can remove the cut-off $\bar\chi_m^{1/2}$ 
by using Theorem~\ref{thm:lap-bounds}.
Hence we are done.
\end{proof}

\subsection{Applications}

Finally we prove Corollaries~\ref{cor:Limiting-Absorption-Principle-Stark}, 
\ref{cor:RC-bound-real} and \ref{cor:Sommerfeld-unique-result}
as applications of Theorems~\ref{thm:priori-decay-b_0}, 
\ref{thm:lap-bounds} and \ref{thm:RC-bound}.

\subsubsection{LAP}

\begin{proof}[Proof of Corollary~\ref{cor:Limiting-Absorption-Principle-Stark}]
Let $s>1/2$ and $\epsilon\in(0,\min\set{\beta_c,s-1/2})$  as in the
assertion. Let $s'=s-\epsilon$.
Decompose for $n\in\mathbb N_0$ and $z,z'\in I_+$ as
\begin{equation}\label{eq:1812121815}
\begin{split}
R(z) - R(z') 
&= 
\chi_nR(z)\chi_n - \chi_nR(z')\chi_n 
\\&\phantom{{}={}}{}
+ \bigl( R(z) - \chi_nR(z)\chi_n \bigr) - \left( R(z') - \chi_nR(z')\chi_n \right).
\end{split}
\end{equation}
We estimate terms on the right-hand side of \eqref{eq:1812121815} as follows.
By Theorem~\ref{thm:lap-bounds} 
we can estimate uniformly in $n\in\mathbb N_0$ and $z, z'\in I_+$ as
\begin{equation} \label{eq:1812121817}
\begin{split}
&\| R(z) - \chi_nR(z)\chi_n \|_{\vL(L^2_s, L^2_{-s})} 
\\&\le 
\| f^{-s}R(z)\bar\chi_nf^{-s} \|_{\vL(L^2)} 
+ \| f^{-s}\bar\chi_nR(z)\chi_nf^{-s} \|_{\vL(L^2)} 
\\&\le 
C_12^{-(s-s')n}=C_12^{-\epsilon n},
\end{split}
\end{equation}
and, similarly, 
\begin{equation} \label{eq:1812121819}
\| R(z') - \chi_nR(z')\chi_n \|_{\vL(L^2_s, L^2_{-s})} 
\le 
C_22^{-(s-s')n}=C_22^{-\epsilon n}.
\end{equation}
As for the first and second terms of \eqref{eq:1812121815}, 
noting  $\overline{a_{\bar z}}=a_z$ and 
\begin{equation} \label{eq:1812121820}
\mathrm i[H, \chi_{n+1}] 
= \mathop{\mathrm{Re}}(\chi_{n+1}'p^f) 
= \mathop{\mathrm{Re}}(\chi_{n+1}'A),
\end{equation}
we can rewrite them as 
\begin{align*}
&\chi_nR(z)\chi_n - \chi_nR(z')\chi_n 
\\&= 
\chi_nR(z)\bigl\{ \chi_{n+1}(H-z') - (H-z)\chi_{n+1} \bigr\}R(z')\chi_n 
\\&= 
\tfrac{\mathrm i}{2}\chi_nR(z)\chi_{n+1}'(A-a_{z'})R(z')\chi_n 
+ \tfrac{\mathrm i}{2}\chi_nR(z)(A+a_{\bar z})^*\chi_{n+1}'R(z')\chi_n
\\&\phantom{{}={}}{}
-\tfrac{\mathrm i}{2}\chi_nR(z)(a_z - a_{z'})\chi_{n+1}'R(z')\chi_n 
+(z-z')\chi_nR(z)\chi_{m}R(z')\chi_n 
\\&\phantom{{}={}}{}
-(z-z')\chi_nR(z)\chi_{m,n+1}(a_z+a_{z'})^{-1}(A-a_{z'})R(z')\chi_n 
\\&\phantom{{}={}}{}
+(z-z')\chi_nR(z)(A+a_{\bar z})^*\chi_{m,n+1}(a_z+a_{z'})^{-1}R(z')\chi_n 
\\&\phantom{{}={}}{}
-(z-z')\chi_nR(z)\bigl[A,\chi_{m,n+1}(a_z+a_{z'})^{-1}\bigr]R(z')\chi_n .
\end{align*}
Here $m\in\mathbb N_0$ is chosen so that $(a_z+a_{z'})^{-1}$
is non-singular on $\mathop{\mathrm{supp}}\bar\chi_m$. 
Then by Theorems~\ref{thm:lap-bounds} and \ref{thm:RC-bound}
we have uniformly in $n\in\mathbb N_0$ and $z, z' \in I_+$
\begin{equation}
\|\chi_nR(z)\chi_n - \chi_nR(z')\chi_n\|_{\vL(L^2_s, L^2_{-s})}
\le 
C_32^{-\epsilon n}
+C_32^{(1-\epsilon)n}|z-z'|.
\label{eq:1812121827}
\end{equation}
By \eqref{eq:1812121815}, \eqref{eq:1812121817}, \eqref{eq:1812121819} 
and \eqref{eq:1812121827}, we obtain uniformly in $n\in\mathbb N_0$ and $z, z' \in I_+$
$$
\| R(z) - R(z') \|_{\vL(L^2_s, L^2_{-s})} 
\le 
C_42^{-\epsilon n}
+C_32^{(1-\epsilon)n}|z-z'|.
$$

Now, if $|z-z'|\le 1$, we choose $n\in\mathbb N_0$ with $2^n \le |z-z'|^{-1} \le 2^{n+1}$, 
and then we obtain
\begin{equation*}
\| R(z) - R(z') \|_{\vL(L^2_s, L^2_{-s})} 
\le 
C_5|z-z'|^\epsilon.
\end{equation*}
The same bound is trivial for $|z-z'|>1$,
and hence the H\"older continuity \eqref{eq:Holder-continuity} 
for $R(z)$ is obtained.
The H\"older continuity \eqref{eq:Holder-continuity} 
for $\widetilde pR(z)$ follows by that for $R(z)$
and the first resolvent equation.

The existence of the limits \eqref{eq:uniform-limit-z-to-lambda} follows 
immediately from \eqref{eq:Holder-continuity}. 
By Theorem~\ref{thm:lap-bounds}
the limits $R(\lambda\pm \mathrm i0)$ and $\widetilde pR(\lambda\pm \mathrm i0)$ 
map into $\vB^*$,
and moreover by density argument these limits extended continuously to maps $\vB\to\vB^*$. 
Hence the assertions are verified.
\end{proof}

\subsubsection{Radiation condition bounds for real spectral parameters}

\begin{proof}[Proof of Corollary~\ref{cor:RC-bound-real}]
The assertion is from Theorem~\ref{thm:RC-bound},
Corollary~\ref{cor:Limiting-Absorption-Principle-Stark} and approximation
arguments. Here we only note the elementary property 
\begin{align*}
  \|\psi\|_{\mathcal B^*}=\sup_{m\in\mathbb N_0} \|\chi_m \psi\|_{\mathcal B^*}
\ \ \text{for }\psi\in\mathcal B^*.
\end{align*}
Hence we are done.
\end{proof}

\subsubsection{Sommerfeld's uniqueness result}

\begin{proof}[Proof of Corollary~\ref{cor:Sommerfeld-unique-result}]
Let $\lambda\in\R$, $\phi\in f^\beta\mathcal B^*$ 
and $\psi\in f^{-\beta}\vB$ with $\beta\in [0,\beta_c)$.
First we let $\phi = R(\lambda+\mathrm i0)\psi$.
Then by Corollaries~\ref{cor:Limiting-Absorption-Principle-Stark} 
and \ref{cor:RC-bound-real} we see that \ref{item:18122818} and \ref{item:18122819} 
of the corollary hold.
Conversely assuming  \ref{item:18122818} and \ref{item:18122819} of
the corollary we let
$$
\phi' = \phi - R(\lambda + \mathrm i0)\psi.
$$
Then by Corollaries~\ref{cor:Limiting-Absorption-Principle-Stark} 
and \ref{cor:RC-bound-real} it follows 
that 
\begin{enumerate}[(1$'$)]
\item\label{item:18122818b}
$(H-\lambda)\phi'=0$ in the distributional sense,
\item\label{item:18122819b}
$\phi'\in f^{\beta}\vB^*$ and $(A- a_+)\phi'\in f^{-\beta}\vB_0^*$.
\end{enumerate}
In addition we have $\phi' \in\vB_0^*$.
To see this  we use functions $\chi_\nu$ as is  \eqref {eq:chimn}, but
considering  now  arbitrary $\nu\in [0,\infty)$. Noting the identity
$$
2\mathop{\mathrm{Im}}\bigl(\chi_\nu(H-\lambda) \bigr) 
= 
(\mathop{\mathrm{Re}}a_+)\chi_\nu' 
+ \mathop{\mathrm{Re}}\bigl( \chi_\nu'(A-a_+) \bigr),
$$
cf.\ \eqref{eq:def-A}, we have the bound
\begin{equation}\label{eq:1812141510}
0 \le 
\langle (\mathop{\mathrm{Re}}a_+)\bar\chi_\nu' \rangle_{\phi'} 
= 
\mathop{\mathrm{Re}}\langle \chi_\nu'(A-a_+) \rangle_{\phi'}.
\end{equation}
Using \ref{item:18122819b} above we deduce that the right-hand side is
bounded as a function of $\nu\geq 0$, leading to the conclusion that
$\phi' \in\vB^*$. Next,  taking  
 the limit $\nu\to\infty$ in \eqref{eq:1812141510}  using again  \ref{item:18122819b}, 
 we indeed obtain $\phi' \in\vB_0^*$.
Then by \ref{item:18122818b} above and Theorem~\ref{thm:priori-decay-b_0} 
it follows that $\phi'=0$, i.e. $\phi = R(\lambda+\mathrm i0)\psi$. 
Hence we are done.
\end{proof}


\begin{thebibliography}{BCDSSW}


\bibitem[Ad]{A} 
T.~Adachi, 
\emph{Local resolvent estimates for $N$-body Stark Hamiltonians}, 
Lett.\ Math.\ Phys.\ \textbf{82} (2007), no.\ 1, 1--8.



\bibitem[AT1]{AT1}
T.~Adachi, H~Tamura,
\emph{Asymptotic completeness for long-range many-particle systems with Stark effect}, 
J. Math.\ Sci.\ Univ.\ Tokyo {\bf 2} (1995), no.\ 1, 76--116.


\bibitem[AT2]{AT2}
T.~Adachi, H~Tamura,
\emph{Asymptotic completeness for long-range many-particle systems with Stark effect. II},
Comm.\ Math.\ Phys.\ {\bf 174} (1996), no.\ 3, 537--559.

\bibitem[AH]{AH}
J.~E.~Avron, I.~W.~Herbst,
\emph{Spectral and scattering theory of Schr{\"o}dinger operators related to the Stark effect},
Comm.\ Math.\ Phys.\ \textbf{52} (1977), no.\ 3, 239--254.

\bibitem[AIIS]{AIIS2}
T. Adachi, K. Itakura, K. Ito, E. Skibsted, 
\emph{Stationary scattering theory for $1$-body Stark Hamiltonians}, 
 arXiv:1905.03539.


\bibitem[BCDSSW]{BCDSSW}
F. Bentosela, R. Carmona, P. Duclos, B. Simon, B. Souillard, R. Weder,
\emph{Schr{\"o}dinger operators with an electric field and random or deterministic potentials},
Comm.\ Math.\ Phys.\ \textbf{88} (1983), no.\ 3, 387--397.

\bibitem[CK]{CK}
M. Christ and A. Kiselev,
\emph{Absolutely continuous spectrum of Stark operators},
Ark.\ Mat.\ \textbf{41} (2003), no.\ 1, 1--33.


\bibitem[He]{H}
I.~W.~Herbst,
\emph{Unitary equivalence of stark Hamiltonians},
Math.\ Z.\ \textbf{155} (1977), no.\ 1, 55--70.


\bibitem[HMS1]{HMS1} 
I.~Herbst, J.S.~M{\o}ller, E.~Skibsted,
\emph{Spectral analysis of $N$-body Stark Hamiltonians}, 
Commun.\ Math.\ Phys.\ {\bf 174}  no.\ 2  (1995), 261--294.
 
 

  \bibitem[HMS2]{HMS2} 
I.~Herbst, J.S.~M{\o}ller, E.~Skibsted, 
\emph{Asymptotic completeness for $N$-body Stark Hamiltonians}, 
Comm.\ Math.\ Phys.\ {\bf 174} (1996), no.\ 3, 509--535. 



\bibitem[It1]{I1}
K.~Itakura,
\emph{Rellich's theorem for spherically symmetric repulsive Hamiltonians},
Math.\ Z.\ {\bf 291} (2019), no.\ 3, 1435--1449.

\bibitem[It2]{I2}
K. Itakura,
\emph{Limiting absorption principle and radiation condition for repulsive Hamiltonians},
to appear in Funkcial.\ Ekvac.

\bibitem[IS]{IS} 
K.~Ito, E.~Skibsted, 
\emph{Spectral theory on manifolds}, 
 arXiv:1602.07488, submitted.

\bibitem[Ki]{K}
A. Kiselev,
\emph{Absolutely continuous spectrum of perturbed Stark operators},
Trans.\ Amer.\ Math.\ Soc.\ \textbf{352} (2000), no.\ 1, 243--256.


\bibitem[Mo]{Mo}
E.~Mourre, 
\emph{Absence of singular continuous spectrum for certain selfadjoint operators}, 
Comm.\ Math.\ Phys.\ \textbf{78} (1980/81), no.~3, 391--408.

\bibitem[NP]{NP}
S.~N.~Naboko and A.~B.~Pushnitski,
\emph{On the embedded eigenvalues and dense point spectrum of the Stark-like Hamiltonians},
Math.\ Nachr.\ \textbf{183} (1997), no.\ 1, 185--200.



\bibitem[RS]{RS} 
M.~Reed, B.~Simon, 
\emph{Methods of modern mathematical physics. II. Fourier analysis, self-adjointness},
Academic Press, New York-London, 1975. 




  \bibitem[Sa]{Sa}
J. Sahbani,
\emph{On the absolutely continuous spectrum of Stark Hamiltonians},
J. Math.\ Phys.\ \textbf{41} (2000), no.\ 12, 8006--8015.




\bibitem[Si]{Si}
I.~M.~Sigal,
\emph{Stark effect in multielectron systems: nonexistence of bound states},
Comm.\ Math.\ Phys.\ {\bf 122} (1989), no.\ 1, 1--22. 

\bibitem[Sk]{Sk} 
E. Skibsted,  
\emph{Absolute spectral continuity for $N$-body Stark Hamiltonians},  
Ann.\ Inst.\ H. Poincar\'e Phys.\ Th\'eor.\ {\bf 61} (1994), no.\ 2 223--243. 





\bibitem[Ta1]{Ta1}
H.~Tamura,
\emph{Spectral and scattering theory for $3$-particle Hamiltonian with Stark effect:
nonexistence of bound states and resolvent estimate}, 
Osaka J. Math.\ \textbf{30} (1993), no.\ 1, 29--55.

\bibitem[Ta2]{Ta2}
H.~Tamura,
\emph{Spectral analysis for $N$-particle systems with Stark effect: nonexistence of bound states
and principle of limiting absorption},
J. Math.\ Soc.\ Japan \textbf{46} (1994), no.\ 3, 427--448.

\bibitem[Ti]{Ti}
E.C.~Titchmarsh,
\emph{Eigenfunction expansions associated with second-order differential equations. Vol. 2}, 
Oxford, at the Clarendon Press 1958. 

\bibitem[Ya1]{Y1} 
K. Yajima, 
\emph{Spectral and scattering theory for Schr{\"o}dinger operators with Stark effect}, 
J. Fac.\ Sci.\ Univ.\ Tokyo Sect.\ IA Math.\ \textbf{26} no.\ 3 (1979), 377--390.

\bibitem[Ya2]{Y2} 
K. Yajima, 
\emph{Spectral and scattering theory for Schr{\"o}dinger operators with Stark effect, II},
J. Fac.\ Sci.\ Univ.\ Tokyo Sect.\ IA Math.\ \textbf{28}  no.\ 1 (1981), 1--15.

\bibitem[Wh]{W}
D. A. W. White,
\emph{The Stark effect and long range scattering in two Hilbert spaces},
Indiana Univ.\ Math.\ J. {\bf 39} (1990), no.\ 2, 517--546. 

\bibitem[Wo]{Wo}
T. Wolff,
\emph{Recent work on sharp estimates in second-order elliptic unique continuation problems},
J. Geom.\ Anal.\ {\bf 3} no. 6 (1993), 621--650.

\end{thebibliography}
\end{document}